\documentclass[preprint,superscriptaddress,nofootinbib,tightenlines]{revtex4}
\usepackage{graphicx}% Include figure files
\usepackage{epstopdf}
\usepackage{amssymb}
\newcommand{\be}{\begin{equation}}
\newcommand{\ee}{\end{equation}}
\usepackage{amsmath}
\usepackage{amsfonts}
\usepackage{bm}
\usepackage{color}
\usepackage{subfigure}
\usepackage{amsthm}
\usepackage[latin1]{inputenc}
\usepackage{tikz}
\usetikzlibrary{decorations.pathreplacing}
\usepackage{comment}
\usepackage{tkz-graph}
\usepackage{mathtools}
\usepackage{centernot}

\newcommand{\Tr}{\textrm{Tr}}

\newcommand{\aeq}{\underset{O(\epsilon)}{\approx}}
\newcommand{\aeqn}{\underset{O(n\epsilon)}{\approx}}
\newcommand{\aeqm}{\underset{O(m\epsilon)}{\approx}}
\newcommand{\pfun}{\mathrel{\ooalign{\hfil$\mapstochar$\hfil\cr$\to$\cr}}}

\theoremstyle{definition}

\newtheorem{lem}{Lemma}
\newtheorem{thm}{Theorem}
\newtheorem{defi}{Definition}

\newtheorem{prop}{Proposition}

\usepackage[colorlinks = true]{hyperref}
\usetikzlibrary{shapes.misc}
\tikzset{cross/.style={cross out, draw=black, minimum size=2*(#1-\pgflinewidth), inner sep=0pt, outer sep=0pt},
cross/.default={7pt}}

\begin{document}
\title{Markovian Marginals}
\author{Isaac H. Kim}
\affiliation{IBM T.J. Watson Research Center, Yorktown Heights, NY 10598, USA}
\affiliation{Perimeter Institute for Theoretical Physics, Waterloo ON N2L 2Y5, Canada}
\affiliation{Institute for Quantum Computing, University of Waterloo, Waterloo ON N2L 3G1, Canada}

\date{\today}

\begin{abstract}
We introduce a class of so called Markovian marginals, which gives a natural framework for constructing solutions to the quantum marginal problem. We consider a set of marginals that possess a certain internal quantum Markov chain structure. If they are equipped with such a structure and are locally consistent on their overlapping supports, there exists a global state that is consistent with all the marginals. The proof is constructive, and relies on a reduction of the marginal problem to a certain combinatorial problem. By employing an entanglement entropy scaling law, we give a physical argument that the requisite structure exists in any states with finite correlation lengths. This includes topologically ordered states as well as finite temperature Gibbs states.
\end{abstract}

\maketitle

\section{Introduction}
Many of the great challenges in modern physics lies on the studies of strongly interacting quantum many-body systems. The main difficulty in these studies is often attributed to the curse of dimensionality. The Hilbert space dimension of the many-body system grows exponentially with the system size, and the number of parameters becomes quickly unmanagable. Methods such as exact diagonalization become unwieldy in those regimes. Then the question is whether one can make reliable approximations to make the problem more tractable.

The fact that certain variational methods, e.g., the density matrix renormalization group(DMRG),\cite{White1992} work so well in practice can be attributed to the fact that the underlying variational ansatz\cite{Fannes1992,DPG2006} can reliably approximate the ground states of physical states with a moderate number of parameters\cite{Hastings2007,Scholz2016}; see also Ref.\cite{Landau2013} for a provably efficient algorithm. Unfortunately, much less is known in higher dimensions. What is clear is that, unlike in one dimension, simply having an area-like upper bound on the entanglement entropy is insufficient to guarantee the existence of an efficient classical description.\cite{Ge2016} Then a natural question is whether there is a different structure that allows us to make progress. It is known that low-temperature states of the systems that obey a reasonable assumption on the density of states can be approximated well by a projected entangled pair states.\cite{Molnar2015} Unfortunately, this in itself does not imply that we can numerically study these systems, because contracting the tensor network is in general computationally hard.\cite{Schuch2007} This means that approximations are necessary, and formulating a condition under which one can approximately contract the network is an ongoing area of research; see \cite{Schwarz2016} for a recent progress in this direction. On the other hand, certain tensor networks can be contracted provably efficiently.\cite{Vidal2008} Here the problem is reversed; many physical states can be shown to be described by such an ansatz in a case-by-case basis\cite{Aguado2008,Koenig2009,Swingle2014} but few formal results are known.

Motivated by this state of affairs, here we propose a completely different direction to study such strongly interacting systems, by providing a certain solution to the quantum marginal problem. In the quantum marginal problem, one is given a set of marginals, i.e., reduced density matrices. Then the question is whether there exists some global state such that the reduced density matrices of this global state is equal to the given set of marginals. When such a global state exists, the marginals are said to be consistent. In its most general form, this problem is QMA-complete, and is unlikely to admit an efficient solution.\cite{Liu2006,Liu2007,Wei2010} Our solution can certify the existence of a global state, but cannot rule out such a possibility. Also, while the condition can be verified efficiently, the set of states that obeys our condition is nonconvex. As such, our result is not in violation of any of the known results.

Our solution is physically relevant, in a sense that the condition that guarantees the consistency of the marginals follows from a rather mild physical condition. A class of states for which our solution becomes nontrivial includes states that obey the following form of entanglement entropy scaling law:
\begin{equation}
S(\rho^A) = \alpha_0 l^D + \alpha_1 l^{D-1} + \cdots, \label{eq:EEscaling}
\end{equation}
where $\alpha_i$ are some constants, $D$ is the number of spatial dimensions, $l$ is the lengthscale of some subsystem $A$ and $S(\rho^A)$ is the von Neumann entropy of a state $\rho^A$.\cite{Hamma2005,Levin2006,Kitaev2006,Eisert2010}  To be more precise, we demand certain linear combinations of these entropies to be sufficiently small, and Eq.\ref{eq:EEscaling} fulfills our demand. While there is a physical argument that all quantum many-body systems with a spectral gap should obey such a relation,\cite{Grover2010} we do not necessarily advocate that point of view. Our stance is merely that our solution is applicable to a wide range of models that have appeared in the literature. The evidence for the universality of Eq.\ref{eq:EEscaling} is surely an encouraging sign for the prospect of our approach, but whether it is indeed true or not remains to be seen.

Our solution is flexible and nontrivial. The class of states that obeys our condition includes so called topologically ordered states\cite{Bravyi2006} as well as highly mixed states. Furthermore, the marginals are allowed to be supported on regions that overlap with each other. We also allow the marginals to be approximately consistent on their overlaps, as opposed to being exactly consistent on their overlaps. In those cases, we can certify the approximate consistency. That is, we can show the existence of some state whose reduced density matrices are close to the given marginals with an uniform upper bound on their trace distance. The nature of our solution is thus very different from the solutions to the so called one-body marginal problem.\cite{Klyachko2004,Daftuar2004,Christandl2006} As a side note, we point out that some of our ideas have already appeared in Ref.\cite{Kato2016}. The main difference is that we have built an entire formalism to be able to deal with multipartite systems.

Admittedly though, these strengths come with a price. Our proof relies on a heavy dosage of formalism, which, to the best of the author's knowledge, does not seem to appear elsewhere. However, for two reasons, we believe there are certain benefits in constructing such a formalism. For one thing, doing so illuminates the essential insight behind the proof. The most important observation is that the original problem, in a certain context that we specify later, can be reduced to a certain combinatorial problem. Our formalism makes this reduction explicit. Second, the formalism opens up a possibility to generate special solutions to the quantum marginal problem by solving a combinatorial problem. It is not completely unreasonable to expect that such an approach might prove useful elsewhere, i.e., in quantum chemistry.

We should also mention an important caveat concerning topologically ordered systems, such as the toric code.\cite{Kitaev1997} While our solution is applicable to arbitrarily large subregions of such system, it is not applicable to the entire system. In fact, we present a general no-go argument in Section \ref{section:conclusion}. As such, a sequence of Markovian marginals designed for such states with increasing sizes does not form a family of states with decreasing energy. It is only the energy density that decreases in this sequence. However, the excitations in this construction would be localized near the boundary, and we expect this effect to be negligible if we only ask questions deep inside the bulk.

The setup and the summary of our results can be found in Section \ref{section:summary}. In Section \ref{section:formalism} we develop the formalism. In Section \ref{section:Proof1} and \ref{section:Proof2}, we prove the main results. We end with some comments in Section \ref{section:conclusion}.

\section{Setup and Summary\label{section:summary}}
Markovian marginal -- the central object of this paper  -- is defined in terms of a collection of three data: partition, clusters, and marginals. Partition, as the name suggests, refers to the partition of the physical degrees of freedom into disjoint subsets. Clusters consist of elements of the partitions, and these are the degrees of freedom on which the marginals are defined. The reason why we make the notion of partition and cluster explicit is because we demand each of the marginals to be equipped with a certain structure that is defined in terms of these objects.

It will be convenient to describe the physical space in terms of a graph. Once this is done, we can make the notion of partition and cluster precise. We consider an undirected graph $G=(V,E)$, where $V$ is a set vertices and $E$ is a set of edges. Physically, one should view the vertices to be a collection of physical particles and the edges to be a bookkeeping device that encodes the locality structure of the underlying system. If for two vertices $v_1,v_2 \in V$ there exists $(v_1,v_2)\in E$, we denote this fact as $v_1- v_2.$ Otherwise, $v_1 \centernot{-} v_2.$ We shall denote the set of neighbors as $\mathcal{N}(v)=\{v'| v-v' \}$. The same convention applies to subsets of the vertices. A partition $\mathcal{P}$ of a graph $G$ is a partition of $V$. That is, $\mathcal{P}=\{V_1,\cdots, V_n \}$ such that $V_i\cap V_j = \emptyset$ $\forall i\neq j$ and $\bigcup_{i=1}^{n} V_i = V.$ The elements of $\mathcal{P} $ shall be referred to as \emph{cells}. A cluster, say $A$, is a collection of cells. We shall denote the union of its elements as $\bar{A}\subset V$.

The physical Hilbert space is a tensor product of finite-dimensional Hilbert spaces labeled by $v\in V$. The degrees of freedom in $\Lambda \subset V$ is defined on the Hilbert space $\mathcal{H}_{\Lambda} = \bigotimes_{v\in \Lambda} \mathcal{H}_v$, where $\dim (\mathcal{H}_v) = d < \infty$. Its algebra of observables is $\mathcal{A}_{\Lambda} = \mathcal{B}(\mathcal{H}_{\Lambda})$. As usual, we define the local algebra of observables as
\begin{equation}
\mathcal{A}_{\text{loc}} = \bigcup_{\Lambda \subset V} \mathcal{A}_{\Lambda}.
\end{equation}
An observable $O$ is supported on $\Lambda$ if $O\in \mathcal{A}_{\Lambda}$. A smallest $\Lambda$ such that $O$ is supported on $\Lambda$ is called as the support of $O$. The space of states, i.e., positive linear functionls of norm $1$, is denoted as $\mathcal{D}_{\Lambda}.$ For our work, it will be useful to define the set of local states:
\begin{equation}
\mathcal{D}_{\text{loc}} = \bigcup_{\Lambda \subset V} \mathcal{D}_{\Lambda}.
\end{equation}
Similar to the observables, we say that a state $\rho$ is supported on $\Lambda$ if $\rho \in \mathcal{D}_{\Lambda}$. Also, if $\rho \in \mathcal{D}_{\Lambda}$, $\Lambda$ is said to be the support of $\rho.$ From now on, we shall specify the support of every state by affixing it in the superscript, e.g., $\rho^{\Lambda}.$  We shall also use completely positive trace preserving(CPTP) maps.\cite{Choi1975} The domain and the codomain of these maps shall appear in the subscript and the superscript respectively. For instance, the domain and the codomain of a CPTP map $\Phi_A^{A'}$ is $A$ and $A'$.  Certain CPTP maps can be represented in a special form:
\begin{equation}
\Phi_A^{A'} = I_B \otimes \Phi_{\tilde{A}}^{\tilde{A}'},
\end{equation}
where $I_B$ is the identity superoperator acting on $\mathcal{A}_B$.
In those cases, we shall say that $\Phi_A^{A'}$ is supported on $\tilde{A}\bigcup \tilde{A}'$. Sometimes, the fact that we are taking a union of two sets will be obvious from the context. In those cases, we will suppress the union symbol. For example, $AB$ should be read as $A\cup B.$

One of the constraints we impose on the marginals is formulated in terms of the von Neumann entropy, $S(\rho) = - \Tr(\rho \log \rho)$, where $\Tr(\cdot)$ is the trace of an operator. Specifically, it involves a linear combination of von Neumann entropy that is known as the conditional quantum mutual information:
\begin{equation}
I(A:C|B)_{\rho} = S(\rho^{AB}) + S(\rho^{BC}) - S(\rho^B) - S(\rho^{ABC}).\label{eq:CMI}
\end{equation}

Now we can define the Markovian marginal in terms of these objects.
\begin{defi}
(Markovian marginal) A Markovian marginal over $G=(V,E)$ is $(\mathcal{P},\mathcal{C}, \mathcal{M})$, where $\mathcal{P}$ is a partition of $V$, $\mathcal{C}$ is a set of clusters, and $\mathcal{M}$ is a set of marginals such that
\begin{itemize}
	\item Cluster condition: $\forall A \in \mathcal{C}$, $A\subset \mathcal{P}$.
	\item Local consistency condition: $\forall \rho^{\bar{A}}, \rho^{\bar{B}} \in \mathcal{M}$, $\Tr_{\bar{A} \setminus \bar{B}}(\rho^{\bar{A}}) = \Tr_{\bar{B} \setminus \bar{A}}(\rho^{\bar{B}})$.
	\item Local Markov condition: $\forall A \in \mathcal{C}$, $\forall a \in A$, $I(a:\bar{A} \setminus(a\bigcup (\mathcal{N}(a) \bigcap \bar{A}))| \mathcal{N}(a) \bigcap \bar{A})_{\rho^{\bar{A}}}=0.$
\end{itemize}
\end{defi}
In words, each of the conditions means the following. The cluster condition means that a cluster is a collection of cells. The local consistency condition means that each of the marginals should be equal to each other once they are restricted to a region on which they overlap. The local Markov condition means that a cell $a$ should obey a certain conditional independence condition within the cluster that it is contained in.\footnote{A tripartite state is loosely said to be conditionally independent if the conditional quantum mutual information is $0.$} This is to be contrasted with the conditional independence relation that is present in Gibbs states of classical statistical mechanics, which can be read as follows\cite{Pearl1988}:
\begin{equation}
I(a: V\setminus (a\cup \mathcal{N}(a)) | \mathcal{N}(a))=0. \label{eq:Markov_graphical_model}
\end{equation}
There are two differences between these two conditions. First, the local Markov condition is formulated in a bounded region of space whereas Eq.\ref{eq:Markov_graphical_model} is formulated in the entire space($V$). Second, the conditioning subsystem(the subsystem $B$ of Eq.\ref{eq:CMI}) is generally a subset of the neighbors of $a$ in the local Markov condition. On the other hand, the conditioning subsystem is chosen to be the entire set of neighbors in Eq.\ref{eq:Markov_graphical_model}.

Markovian marginal, in its present form, is an idealistic construct because one would need infinite precision to verify the constraints. This is clearly unrealistic, which is why we need to introduce an $\epsilon-$Markovian marginal.
\begin{defi}
($\epsilon-$Markovian marginal) An $\epsilon-$Markovian marginal over $G=(V,E)$ is $(\mathcal{P},\mathcal{C}, \mathcal{M})$, where $\mathcal{P}$ is a partition of $V$, $\mathcal{C}$ is a set of clusters, and $\mathcal{M}$ is a set of marginals such that
\begin{itemize}
	\item Cluster condition: $\forall A \in \mathcal{C}$, $A\subset \mathcal{P}$.
	\item Local consistency condition: $\forall \rho^{\bar{A}}, \rho^{\bar{B}} \in \mathcal{M}$, $\|\Tr_{\bar{A} \setminus \bar{B}}(\rho^{\bar{A}}) - \Tr_{\bar{B} \setminus \bar{A}}(\rho^{\bar{B}})\|_1 \leq \epsilon$. If $\bar{A}\subset \bar{B}$, $\rho^{\bar{A}} = \Tr_{\bar{B} \setminus \bar{A}}(\rho^{\bar{B}})$.
	\item Local Markov condition: $\forall A \in \mathcal{C}$, $\forall a \in A$, $I(a:\bar{A} \setminus(a\bigcup (\mathcal{N}(a) \bigcap \bar{A}))| \mathcal{N}(a) \bigcap \bar{A})\leq \epsilon^2.$
\end{itemize}
\end{defi}
Here $\|\cdots \|_1$ is the trace distance.\footnote{The reason why we chose $\epsilon^2$ as opposed $\epsilon$ is to simplify the analysis, as it shall become apparent later. In any case, for a sufficiently small $\epsilon$, the main conclusion of this paper is insensitive to such details.} An $\epsilon-$Markovian marginal is a more realistic construct. Provided that each of the marginals are supported on a finite-dimensional space, by keeping $O(\log \frac{1}{\epsilon})$ bits of precision, the conditions can be certified rigorously. This is a simple fact that follows from the triangle inequality for norms and the continuity of the von Neumann entropy.\cite{Fannes1973} We emphasize that the exact equality in the local consistency condition for the case of $A\subset B$ is not an unrealistic assumption. In such cases, one can simply store the marginal $\rho^{\bar{B}}$ up to a fixed precision. Then one can simply \emph{define} $\rho^{\bar{A}}$ as the reduced density matrix of $\rho^{\bar{B}}$.

By making an appropriate choice of $G, \mathcal{P},$ and $\mathcal{C}$, we shall prove that there exists some state $\sigma$ whose marginals over $\mathcal{C}$ are consistent with $\mathcal{M}$ up to a small error, say $\delta$. Specifically, we shall be able to show that certain $\epsilon-$Markovian marginals are $O(|V|\epsilon)$-consistent, where we define $\delta-$consistency as follows:
\begin{defi}
An $\epsilon-$Markovian marginal $(\mathcal{P},\mathcal{C},\mathcal{M})$ is $\delta-$consistent if $\exists \sigma \geq0$ such that $\forall \rho^{\bar{A}} \in \mathcal{M}$, $\| \rho^{\bar{A}} - \sigma^{\bar{A}}\|_1\leq \delta$.
\end{defi}

\subsection{Main Result: $D=1$}

We state one of our main results and justify its physical relevance by invoking a physical argument in the literature. We begin by specifying the graph. Let $G=(V,E)$, where $V=\{1,\cdots, n \}$ and $E=\{(i,i+1)| i=1, \cdots, n-1 \}$. We assume that $n$ is an even number. We prove the statement later in Section IV.
\begin{thm}
An $\epsilon-$Markovian marginal $(\mathcal{P}, \mathcal{C}, \mathcal{M})$ with $\mathcal{P} = \{[i]| i=1, \cdots, \frac{n}{2} \}$ where $[i]:= \{2i-1, 2i\}$, $\mathcal{C} = \{ \{[i],[i+1] \}| i=1, \cdots, \frac{n}{2}-1 \}$ is $cn\epsilon$-consistent, where $c$ is a constant that is independent of $d$ and $n$.\label{thm:1D}
\end{thm}
The size, i.e., the number of vertices in each clusters, is $O(1)$, which means that each of the marginals is supported on a $d^{O(1)}$-dimensional space. Therefore, the total number of parameters that define these states is $O(nd^{O(1)})$, up to a multiplicative factor that depends on the precision.\footnote{We can choose it to be this factor to be $O(\log n)$ to ensure $1/\text{poly}(n)$ accuracy. This additional factor does not change our conclusion.}

Now let us write down the constraints and discuss their physical relevance. The local Markov conditions are of the following form:
\begin{equation}
\begin{aligned}
I( \{2i-1,2i \} : \{2i+2 \} | \{2i+1 \}) &\leq \epsilon^2 \\
I(\{2i+1,2i+2 \}: \{2i-1 \} | \{2i \}) &\leq \epsilon^2.
\end{aligned}
\end{equation}
This means that we demand $I(A:C|B)\leq \epsilon^2$ for three contiguous subsystems $A,B,$ and $C$ that are sitting side by side on a line. By vieweing the $d-$dimensional particle as a collection of $l$ elementary particles on a line, we can invoke the so called scaling law of entanglement entropy.\cite{Eisert2010} This ``law'' is a conjecture that, the von Neumann entropy of such subsystems in certain states can be decomposed into a sum of local contributions so that each of these local contributions can be canceled out from different choices of subsystems, modulo a correction term that decays exponentially in $l$.\cite{Grover2010} If this conjecture is true, then every local contribution from $S(\rho^{ABC})$, $S(\rho^{BC})$, $S(\rho^B)$, and $S(\rho^{ABC})$ cancels each other out, and what remains is a term that decays exponentially in $l$.\footnote{See Ref.\cite{Kato2016} for a recent progress in proving a similar statement for thermal states.} For states that obey such a scaling law, one can choose $l= \Theta(\log n)$ so that the corresponding Markovian marginal is $\frac{1}{\text{poly}(n)}$-consistent and has $\text{poly}(n)$ number of parameters. Under this circumstance, a Markovian marginal would be an efficient classical description of some state that is approximately consistent with the given marginals.\footnote{Of course, we cannot guarantee the uniqueness of the state. However, if the uniqueness is truly necessary, one can simply choose the state to be the maximum entropy state that is consistent with the given marginals. }

\subsection{Main Result: $D=2$}
We present our second main result, drawing parallels with our first result from time to time. The underlying graph behind this construction is not planar, so it will be more convenient to depict the physical space and then define the graph in terms of this physical space. Consider a bounded region of a triangular lattice superimposed with its dual lattice; see FIG.\ref{fig:physical_space_2d} The vertices of the graph consists of the faces, and an edge exists if and only if the faces corresponding to the vertices are adjacent to each other.
\begin{figure}[h]
\includegraphics[width=3in]{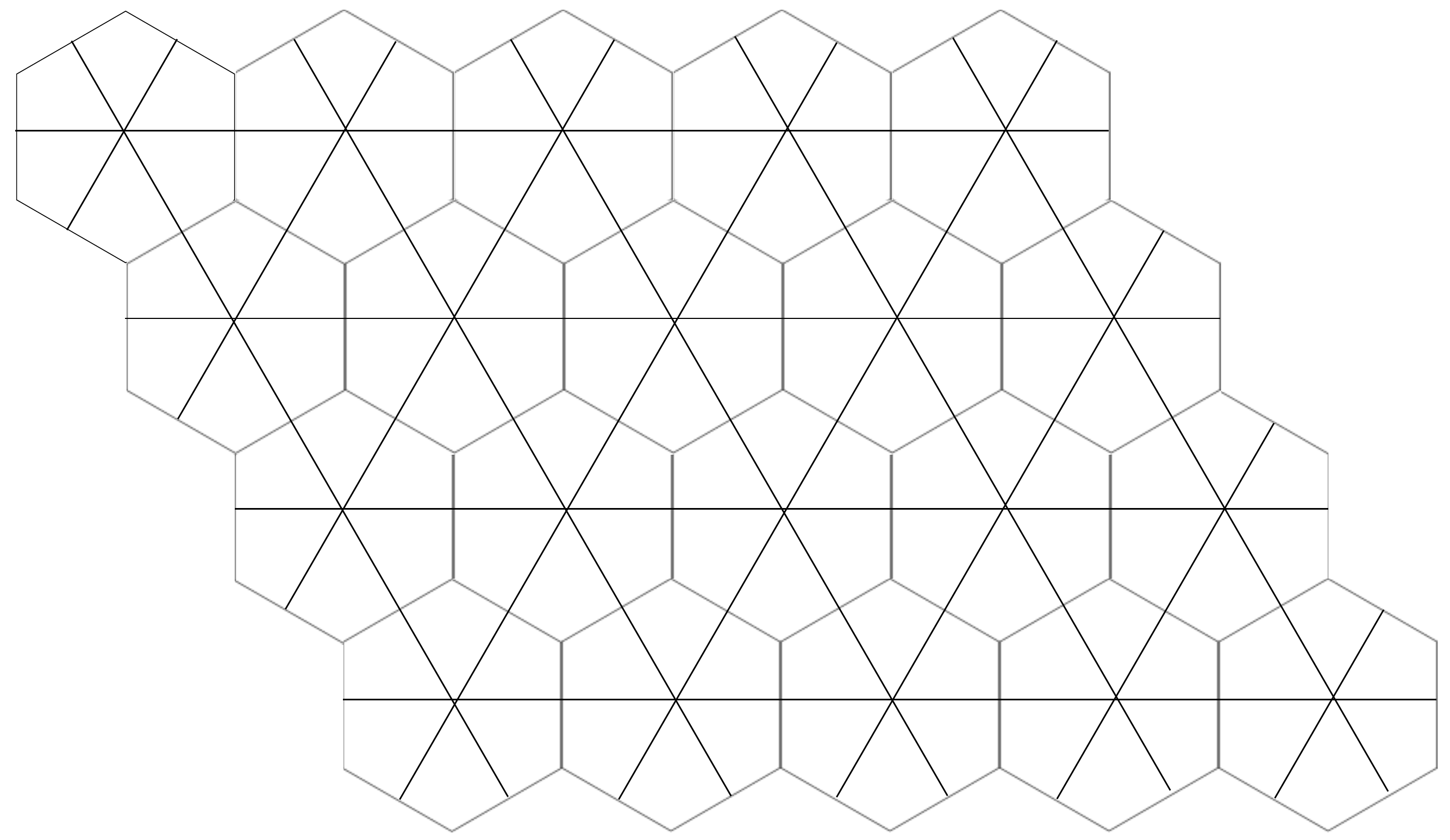}
\caption{Each of the faces can be thought as a partition of the physical space that contains bounded number of physical degrees of freedom. The faces in this diagram corresponds to the vertices of the graph $G=(V,E)$. An edge exists between two vertices if the corresponding faces are adjacent to each other.\label{fig:physical_space_2d} }
\end{figure}

Now we discuss the partition of this graph; see FIG.\ref{fig:partition_2d} The original physical space in FIG.\ref{fig:physical_space_2d} is partitioned into a set of hexagons, and the set of vertices that correspond to the faces of these hexagons are labeled by tuples of integers $[i,j]$, which represents the coordinate of the center of the hexagons. Specifically, $[i,j]$ represents a hexagon at coordinate $(i-\frac{1}{2}j, \frac{\sqrt{3}}{2}j)$.
\begin{figure}[h]
\includegraphics[width=3in]{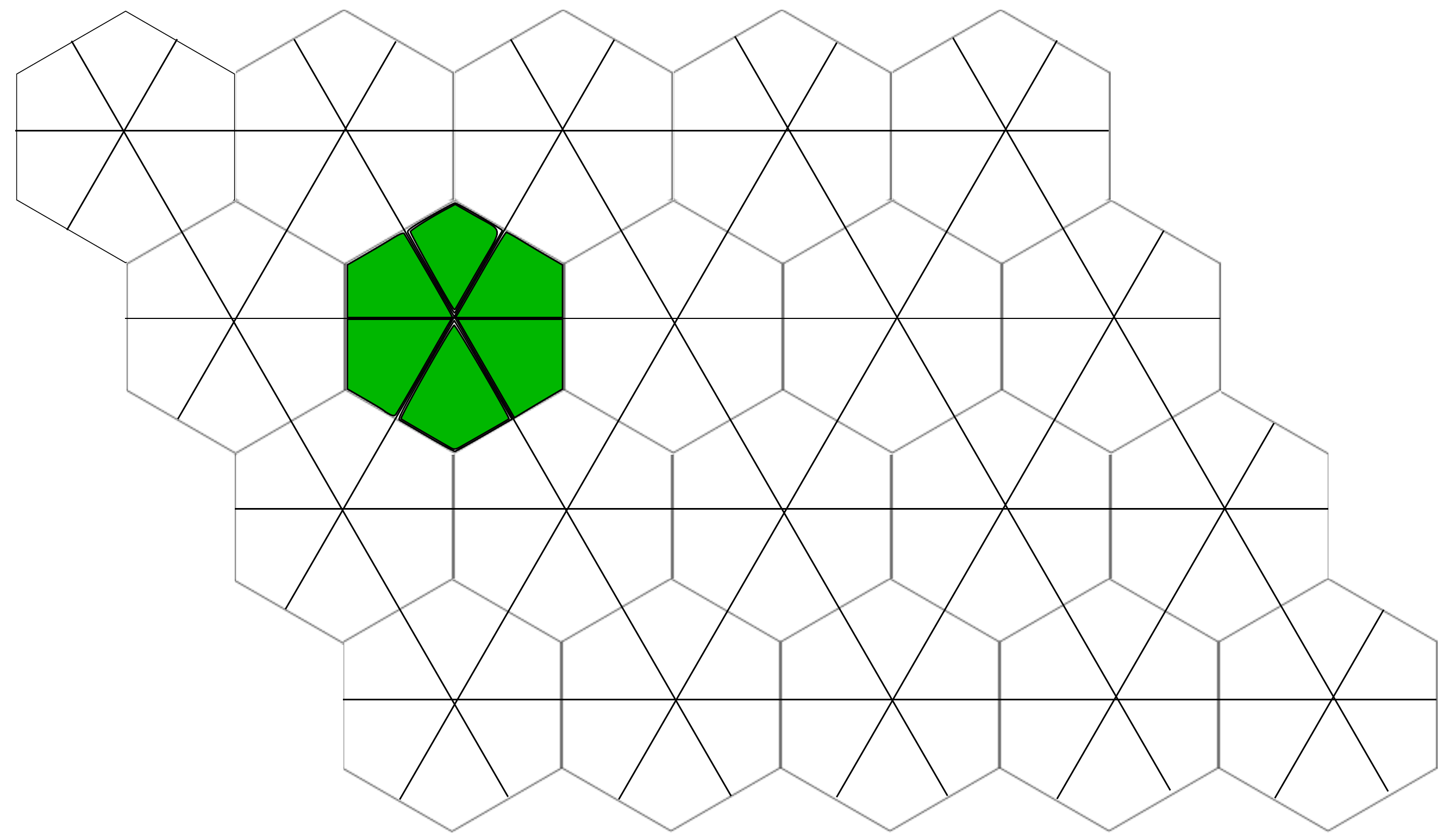}
\caption{The original physical space is partitioned into a set of hexagons. Each of the hexagons are labeled by tuples of integers $[i,j].$ These tuples represent the cells of the partition, and also the location of the center of the hexagon. For reference, we set the center of the hexagon left bottom corner to be $(0,0)$ and the distance between the center of the hexagons to be $1.$ \label{fig:partition_2d}}
\end{figure}

Now we can state our main result.
\begin{thm}
An $\epsilon-$Markovian marginal $(\mathcal{P},\mathcal{C},\mathcal{M})$ with $\mathcal{P}$ described in FIG.\ref{fig:partition_2d}, $\mathcal{C}= \mathcal{C}_3 \bigcup \mathcal{C}_4$, where
\begin{equation}
\begin{aligned}
\mathcal{C}_3 &= \{ \{[i,j], [i+1,j], [i,j+1] \}, \{[i+1,j], [i+1,j+1], [i,j+1] \}| i,j = 1, \cdots, n-1\} \\
\mathcal{C}_4 &= \{\{[i,j], [i+1,j], [i,j+1], [i+1,j+1] \}|i,j = 1, \cdots, n-1 \}
\end{aligned}
\end{equation}
is $cn^2\epsilon$-consistent, where $c$ is independent of $d$ and $n$. \label{thm:2D}
\end{thm}
Again, each clusters contain at most $O(1)$ particles of dimension $d$, which means that each of the marginals is supported on a $d^{O(1)}$-dimensional space. The total number of parameters that define these states is then $O(n^2d^{O(1)})$.

Let us write down the local Markov constraints and see if they are physically reasonable. The conditions from $\mathcal{C}_3$ are of the following form:
\begin{equation}
I(A:C|B) \leq \epsilon^2,
\end{equation}
where $A,B,$ and $C$ are two-dimensional regions that are depicted in FIG.\ref{fig:CMI2D_3}
\begin{figure}[h]
\subfigure[]{\includegraphics[width=1.5in]{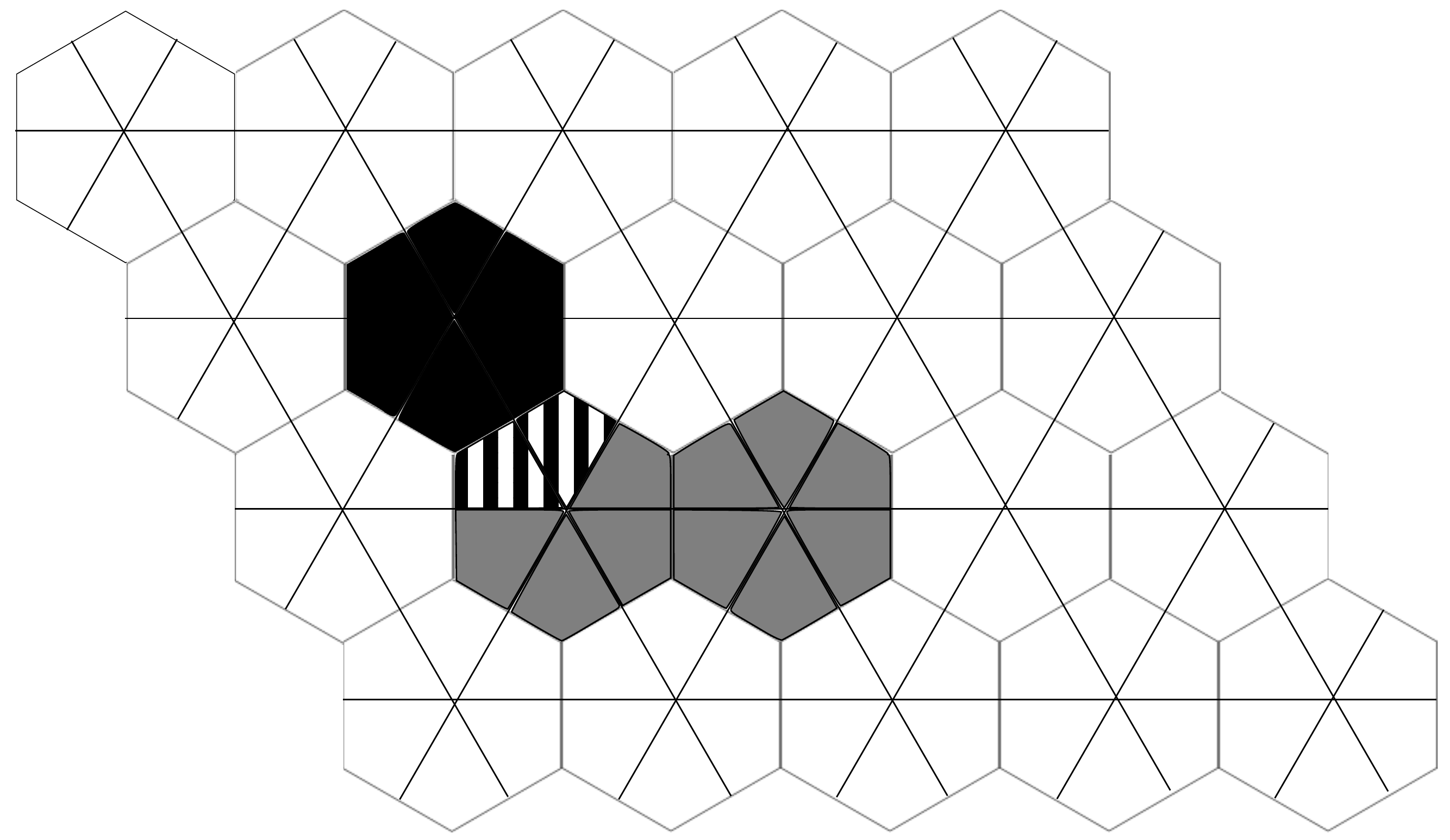}\label{fig:CMI2D_3-1}}
\subfigure[]{\includegraphics[width=1.5in]{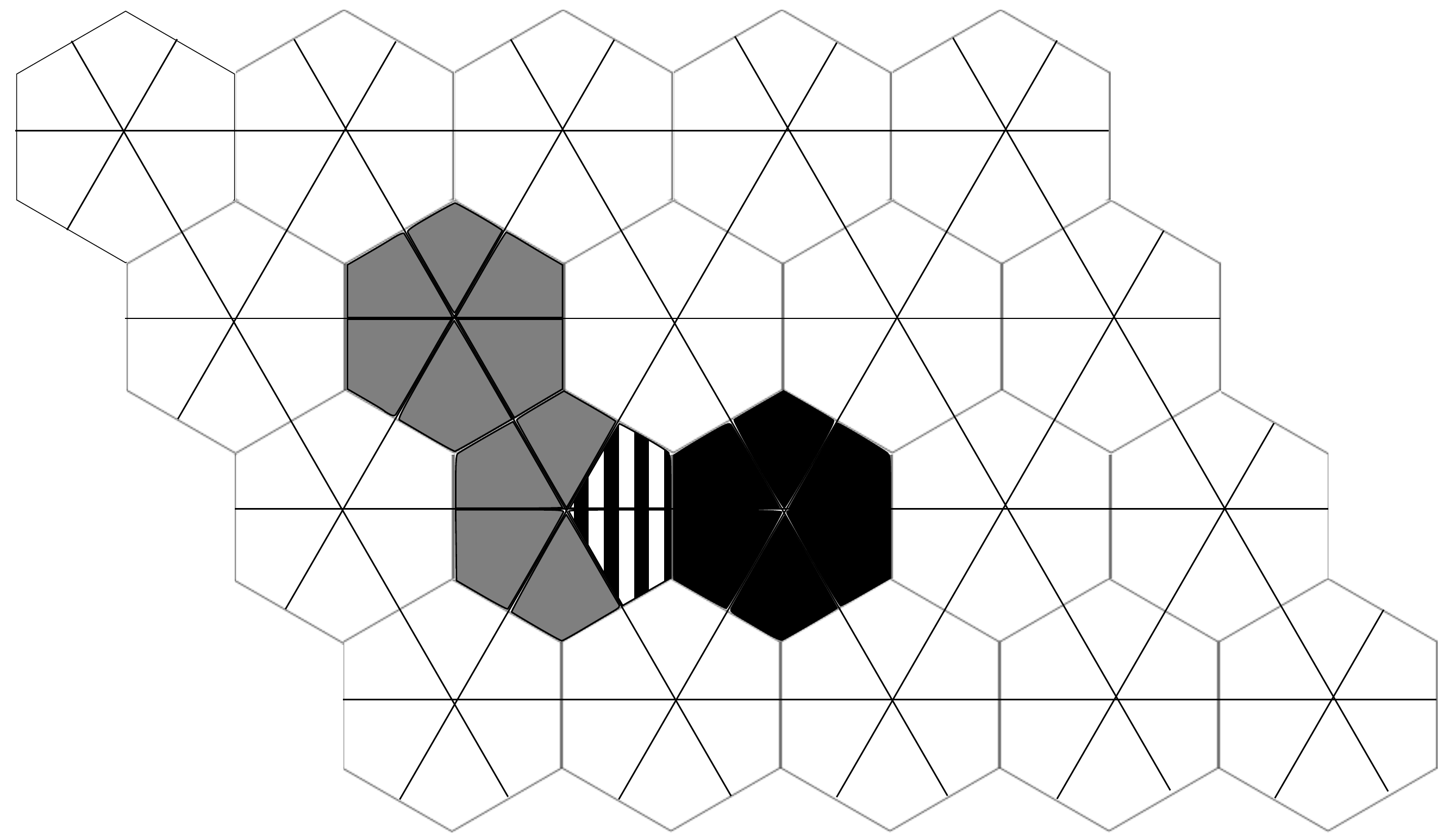}\label{fig:CMI2D_3-2}}
\subfigure[]{\includegraphics[width=1.5in]{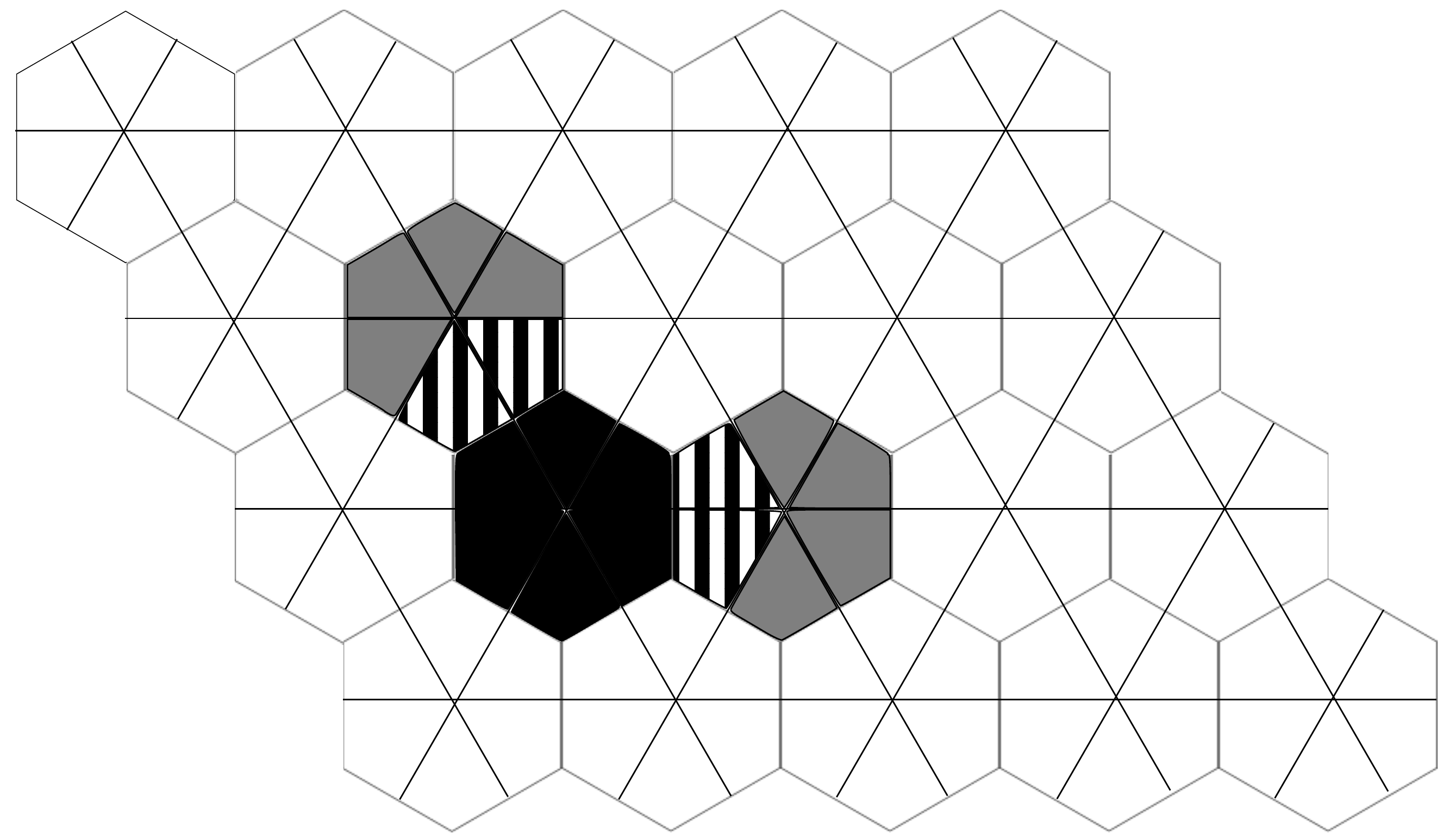}\label{fig:CMI2D_3-3}}

\subfigure[]{\includegraphics[width=1.5in]{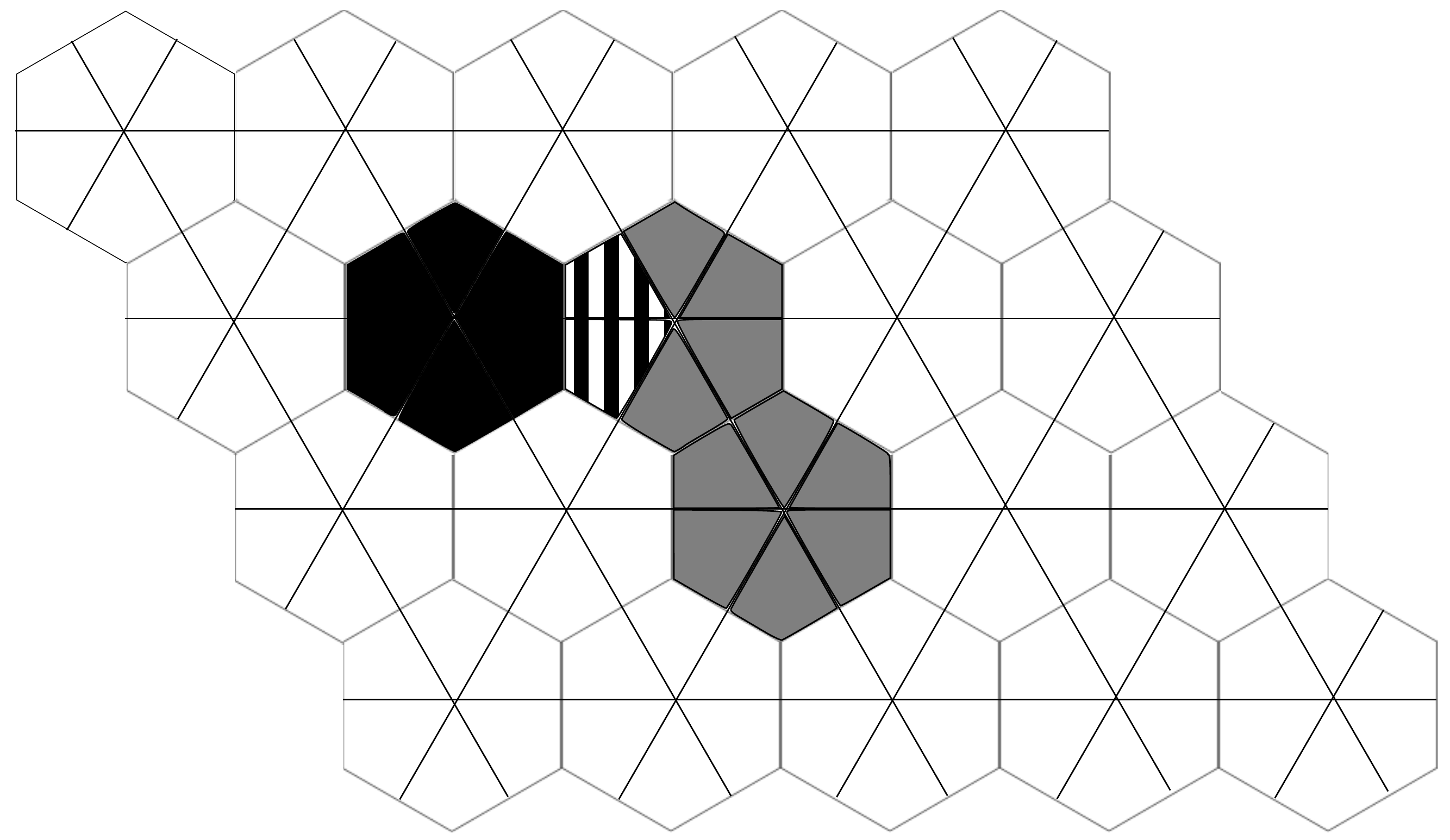}\label{fig:CMI2D_3-4}}
\subfigure[]{\includegraphics[width=1.5in]{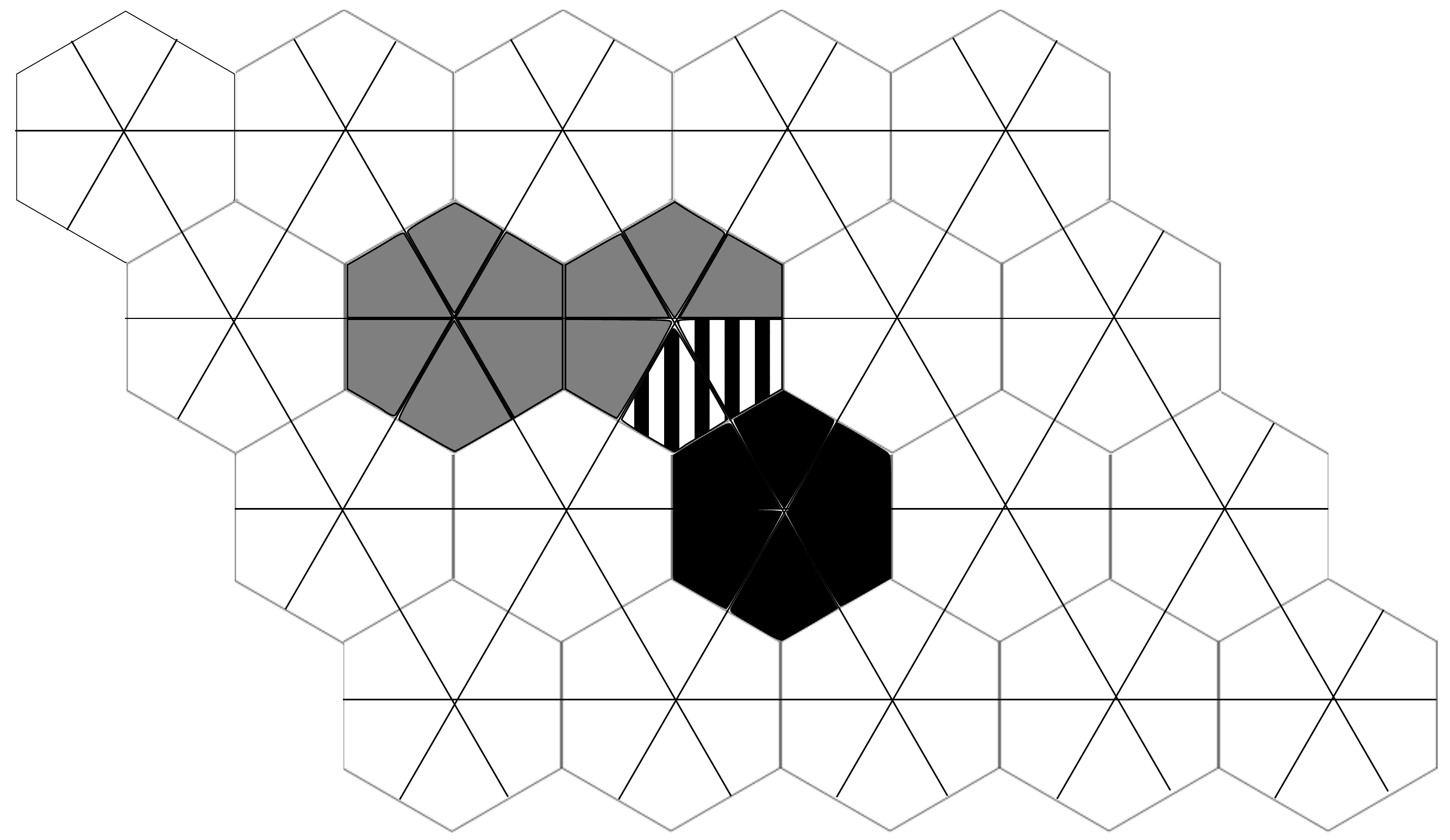}\label{fig:CMI2D_3-5}}
\subfigure[]{\includegraphics[width=1.5in]{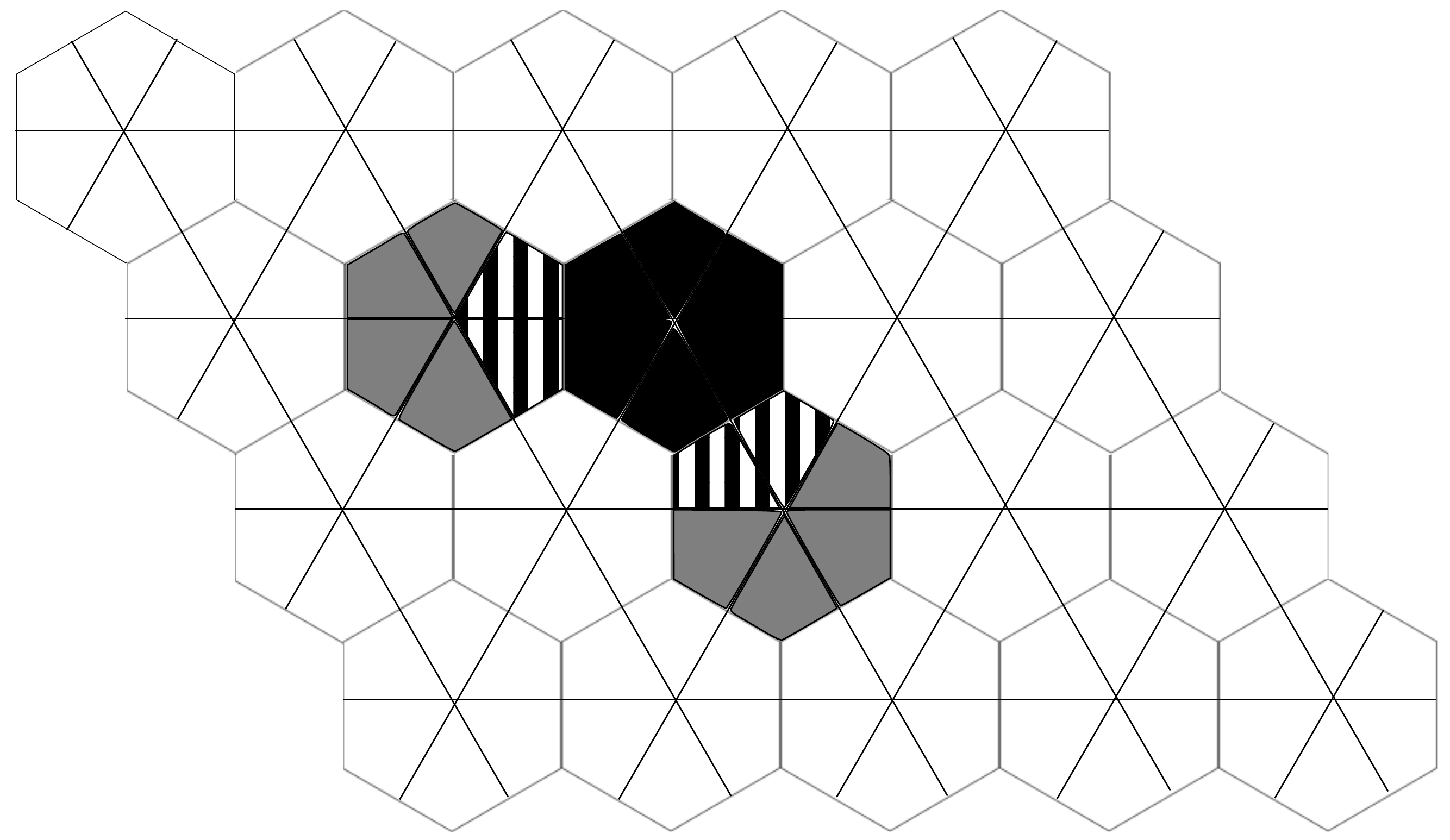}\label{fig:CMI2D_3-6}}
\caption{The dark region is $A$, the striped region is $B$, and the gray region is $C$; see Eq.\ref{eq:CMI}. \label{fig:CMI2D_3}}
\end{figure}
It is interesting to apply the existing formulae for the entanglement entropy and compute the conditional quantum mutual information for a set of subsystems depicted in FIG.\ref{fig:CMI2D_3}. According to a certain physical argument, in ground states of systems that have a spectral gap to the excited states, the entropy of these regions are given by the following expression:\cite{Kitaev2006,Grover2010}
\begin{equation}
S(\rho^A) = \alpha l - n(A)\gamma + e^{-O(l)},\label{eq:TEE}
\end{equation}
where $l$ is the perimeter of a region $A$ and $n(A)$ is the number of connected components.\footnote{It is assumed that each of the connected components is simply connected.} If we view each of the $d$-dimensional particles to be a collection of $O(l^2)$ elementary particles, plugging in Eq.\ref{eq:TEE} leads to the following conclusion:
\begin{equation}
I(A:C|B) = e^{-O(l)}
\end{equation}
for the choices of $A,B,$ and $C$ in FIG.\ref{fig:CMI2D_3}. Therefore, for such states, the local Markov condition that is implied by $\mathcal{C}_3$ would be reasonable so long as $l$ is chosen to be $\Theta(\log \frac{1}{\epsilon})$.

Now let us move on to the local Markov conditions that follow from $\mathcal{C}_4.$ Again the condition reads as $I(A:C|B)\leq\epsilon^2$ for certain choices of $A,B,$ and $C$. All the possibilities are depicted in FIG.\ref{fig:CMI2D_4}.
\begin{figure}[h]
\subfigure[]{\includegraphics[width=1.5in]{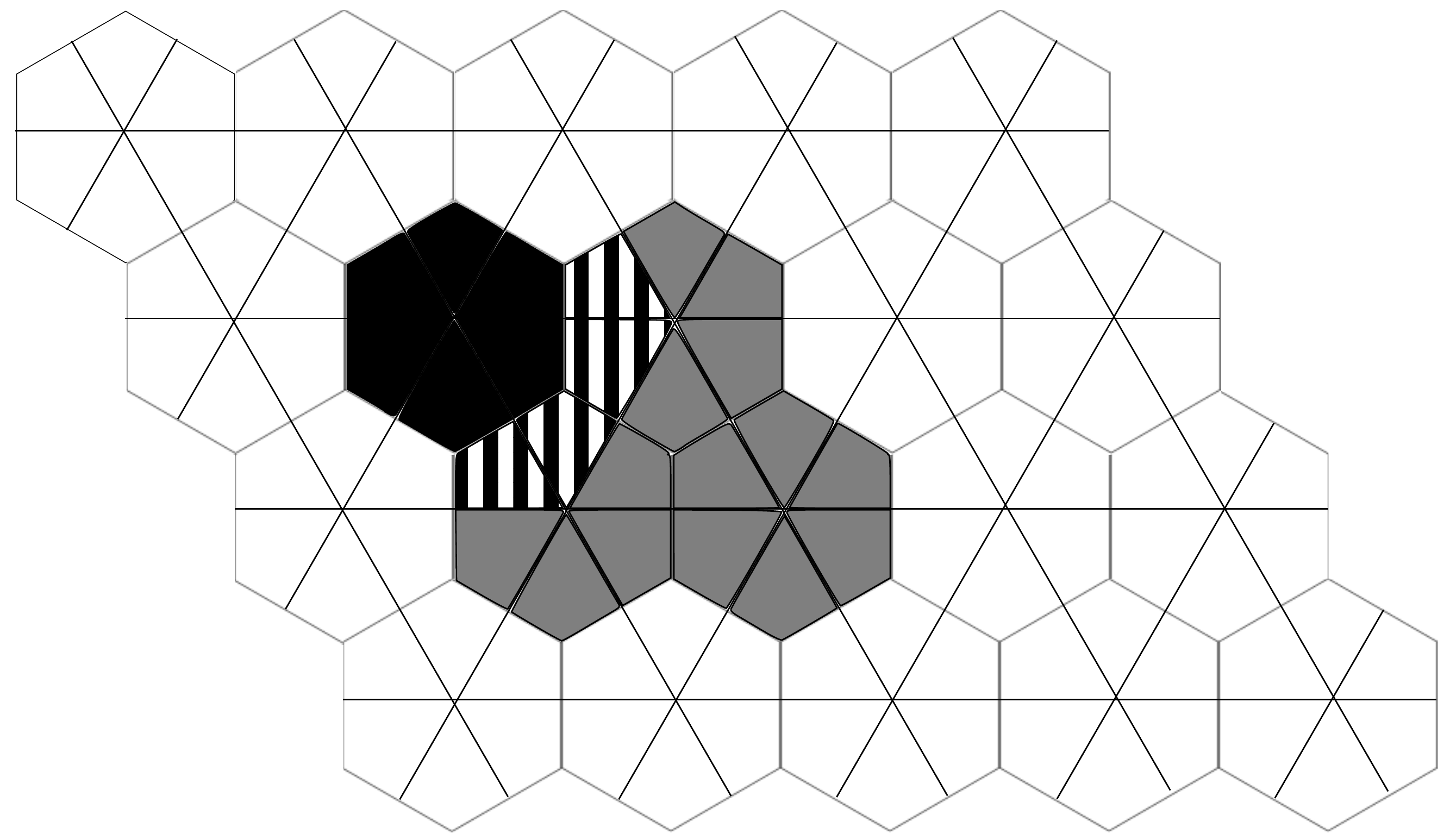}\label{fig:CMI2D_4-1}}
\subfigure[]{\includegraphics[width=1.5in]{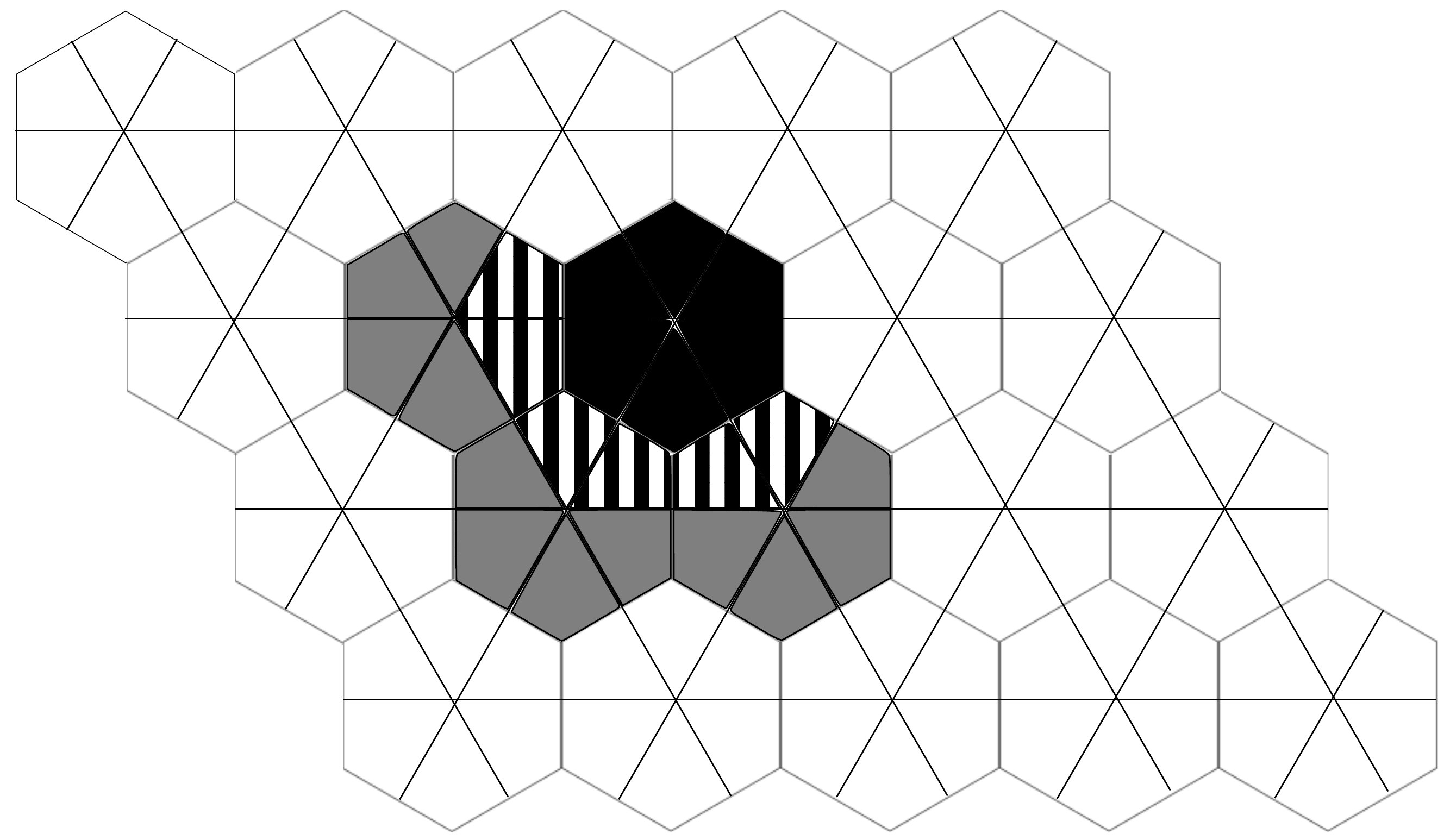}\label{fig:CMI2D_4-2}}
\subfigure[]{\includegraphics[width=1.5in]{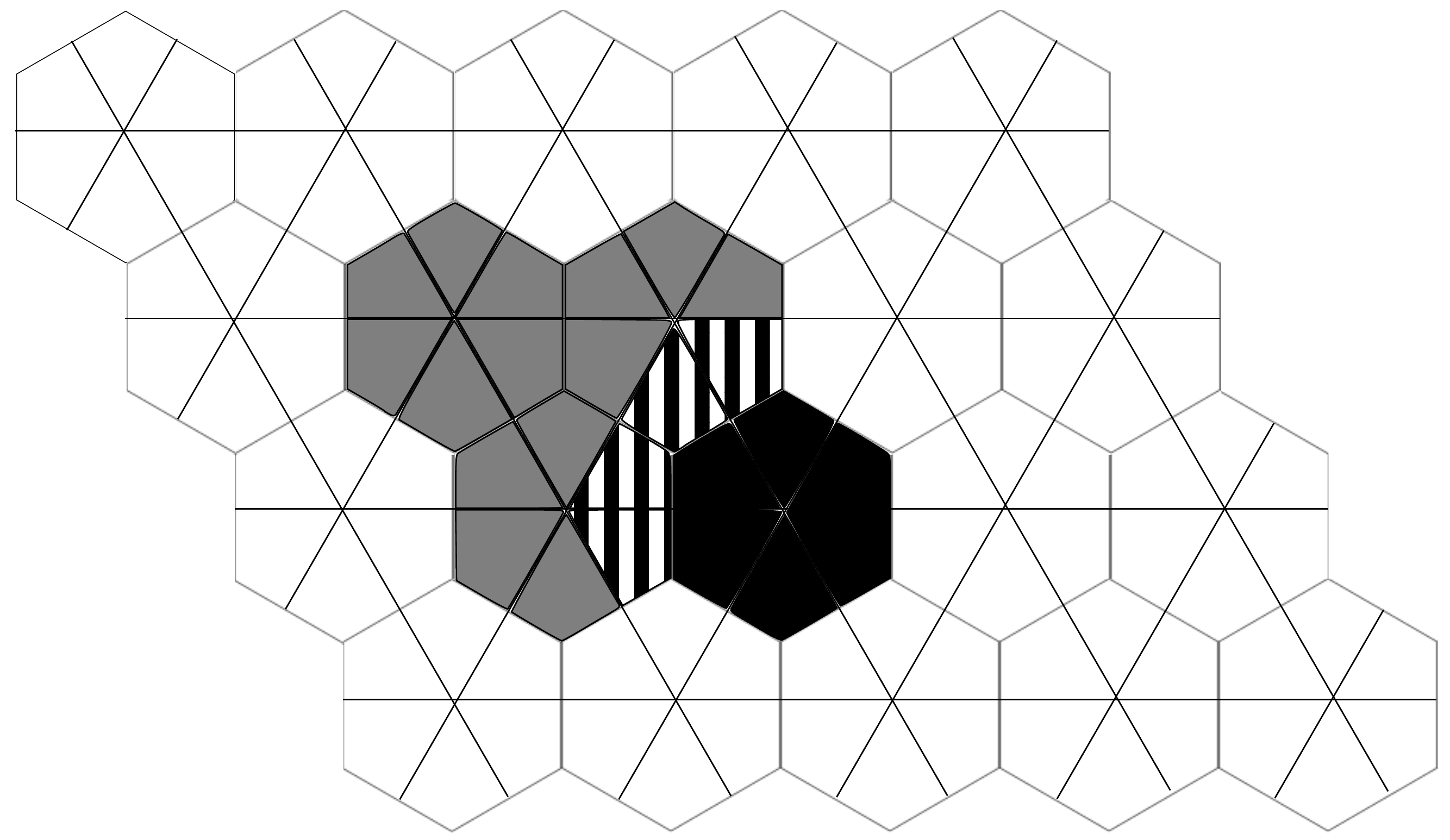}\label{fig:CMI2D_4-3}}
\subfigure[]{\includegraphics[width=1.5in]{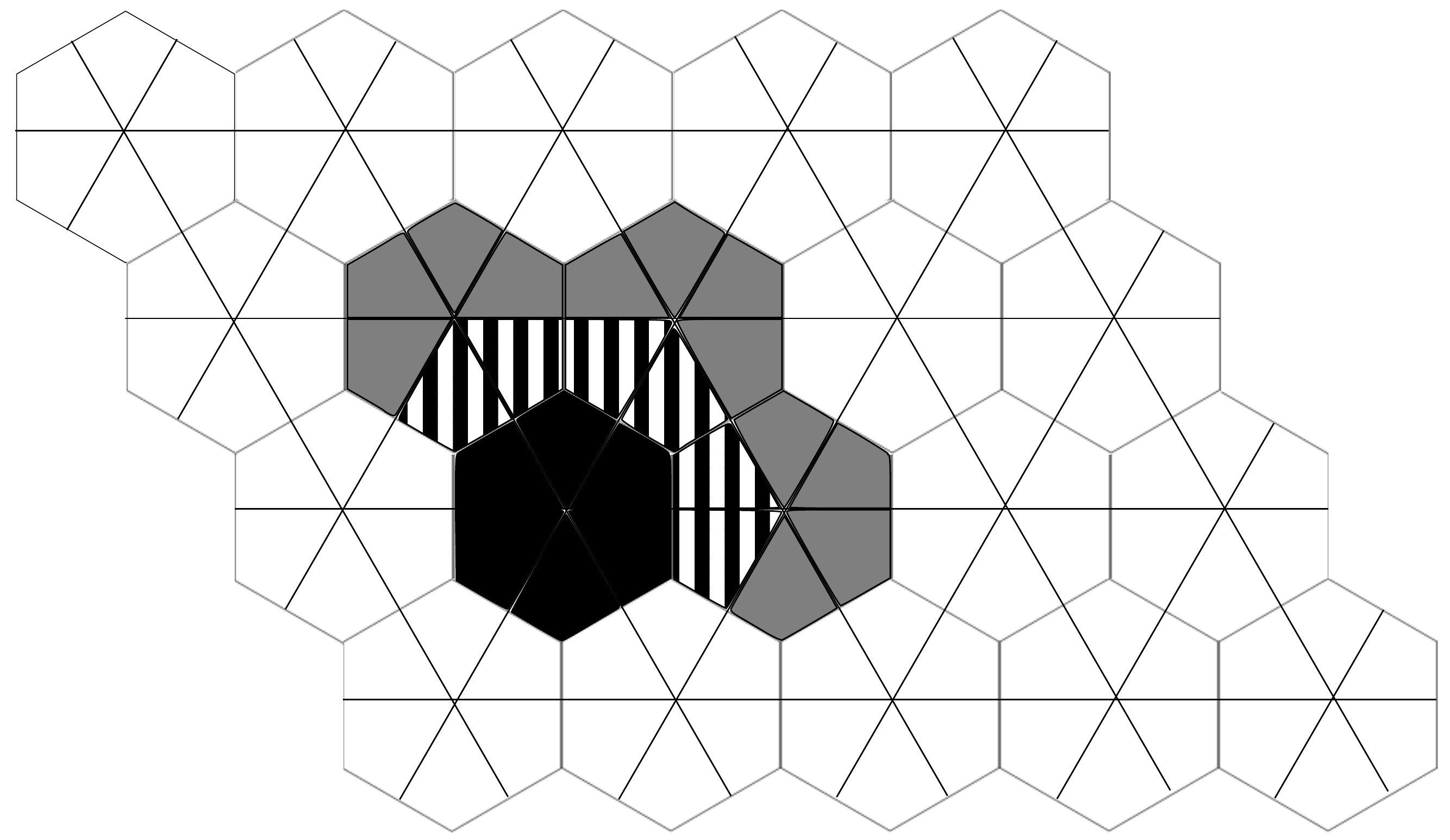}\label{fig:CMI2D_4-4}}
\caption{The dark region is $A$, the striped region is $B$, and the gray region is $C$; see Eq.\ref{eq:CMI}.\label{fig:CMI2D_4}}
\end{figure}
Again, by choosing each of the $d$-dimensional particles to be a collection of $O(l^2)$ elementary particles and plugging in Eq.\ref{eq:TEE}, we conclude that $I(A:C|B)= e^{-O(l)}$ for the choices of subsystems in FIG.\ref{fig:CMI2D_4}.

The conclusion that we can draw from this exercise is that the local Markov condition we demand for the Markovian marginal is a physically reasonable and well-motivated one, at  least for many well-studied models of quantum many-body systems with a spectral gap. Furthermore, by choosing $l=\Theta(\log n)$, the corresponding Markovian marginal becomes $\frac{1}{\text{poly}(n)}$-consistent and has $O(\exp(c \log^2 n))$ number of parameters. (Here $c$ is some numerical constant.) This leads to an exponential reduction on the number of parameters, as well as the computation time for computing the expectation values of local observables, at least compared to exact methods. We prove these claims in the remaining part of the paper.

\section{Formalism\label{section:formalism}}
The main objective of our formalism is to reduce the marginal problem, which is a problem of an algebraic nature, to a combinatorial problem. This procedure can be divided largely into two steps. In the first step, we introduce a family of partial functions that are well-defined for any Markovian marginals. In the second step, we show that these partial functions obey a certain set of identities. At that point, what remains is an application of such identitites, which shall be extensively used in Section IV to prove our main results.

The partial functions are defined in two steps. We first define a family of maps that act on different spaces of states, and then we define two different families of partial functions, $\mathfrak{C}_a: \mathcal{D}_{\text{loc}} \pfun \mathcal{D}_{\text{loc}}$ and $\mathfrak{E}_{a^A}: \mathcal{D}_{\text{loc}} \pfun \mathcal{D}_{\text{loc}}$, which are defined in terms of these linear maps. The linear maps acting on different spaces of states are the so called \emph{universal recovery maps}.\cite{Sutter2015}. By defining a partial function that is constructed out of a collection of such maps, we shall define what we call as \emph{polymorphic extensions}($\mathfrak{E}_{a^A}$) and \emph{polymorphic contractions}($\mathfrak{C}_a$).

The identitites that we derive are identities involving the polymorphic extensions and contractions. We shall refer to these identities as \emph{relations}. There are two types of relations. The type of the first kind is called as \emph{manifest relations}, and these identities hold exactly. The type of a second kind is called as \emph{derived relations}, and they are only guaranteed to be correct up to an error of order $O(\epsilon)$. The first type follows directly from the definition of the partial functions, wheras the second type follows from the local consistency and the local Markov condition. For this reason, the derivation of the manifest relations will be identical for both of our main results. On the other hand, the derivation of the derived relations shall be different. The proof of the manifest relations can be applied to both the $D=1$ and the $D=2$ case, whereas the derived relations are derived separately. The manifest relations are summarized in Table \ref{table:1D} and \ref{table:2D}. The derived relations shall appear in Section \ref{section:Proof1} and \ref{section:Proof2}; see also Table \ref{table:derived_1D}, \ref{table:inheritance}, and \ref{table:row_relations} for the summary. Both Theorem \ref{thm:1D} and \ref{thm:2D} heavily rely on extensive use of these relations.

\subsection{Universal recovery maps}
We first begin by defining the so called universal recovery maps. Recently, a great deal of advance has been made on the structure of tripartite states with a small conditional quantum mutual information.\cite{Fawzi2015,Wilde2015,Sutter2015,Junge2015} A relevant result to our paper is due to Sutter, Fawzi, and Renner:
\begin{thm}
\cite{Sutter2015} There exists a CPTP map $\Phi_B^{BC}: \mathcal{A}_{B} \to \mathcal{A}_{BC}$ such that
\begin{equation}
-2\log F(\rho^{ABC}, I_A \otimes \Phi_B^{BC}(\rho^{AB})) \leq I(A:C|B).
\end{equation}
In particular, $\Phi_B^{BC}$ can be defined only in terms of $\rho^{BC}$, and not $\rho^{ABC}$. \label{thm:universal_recovery}
\end{thm}\footnote{Here $F(\rho, \sigma)= \| \rho^{\frac{1}{2}} \sigma^{\frac{1}{2}}\|_1$ is the fidelity.}
The map $\Phi_B^{BC}$ is called as a universal recovery map from $B$ to $BC.$ Some of its useful properties include $\Phi_{B}^{BC}(\rho^B) = \rho^{BC}$ and the fact that it is norm-nonincreasing. Also, by invoking the standard relation between fidelity and the trace norm, for $I(A:C|B)\leq \epsilon^2$, we can infer that $\|\rho^{ABC} - I_A \otimes \Phi_B^{BC}(\rho^{AB}) \|_1 \leq O(\epsilon)$.

In the context of the Markovian marginal, these universal recovery maps will be defined in terms of the given marginals. Since the state over each clusters are uniquely defined, these maps can be unambiguously defined by declaring the cluster that supports the marginal.
\begin{defi}
For an $\epsilon-$Markovian marginal $(\mathcal{P},\mathcal{C}, \mathcal{M})$, for $A\in \mathcal{C}$, $BC\subset \bar{A}$
\begin{equation}
[\mathcal{R}_A]_B^{BC} := \Phi_B^{BC},
\end{equation}
where $\Phi_B^{BC}$ is a map defined in Theorem \ref{thm:universal_recovery} in terms of a reduced state of $\rho^{\bar{A}}\in \mathcal{M}$ over $BC.$
\end{defi}

\subsection{Polymorphic extensions and contractions}
We can motivate the content of this Section by posing the following question. For a tripartite system $ABC$, does the partial trace operation on $A$ and $B$ commute? The answer is no. The partial trace operation on $\mathcal{B}(\mathcal{H}_A \otimes \mathcal{H}_B \otimes \mathcal{H}_C)$ is formally $\Tr_A \otimes I_{BC}$, where $\Tr_A$ is the trace over the subsystem $A$ and $I_{BC}$ is the identity superoperator over $BC$. The partial trace operation over $B$ is formally $\Tr_B \otimes I_{AC}$. We cannot compose thse two maps because the domain of one map does not match the codomain of the other map. Nevertheless, there is a sense in which their order does not matter. The exact identity is the following:
\begin{equation}
(\Tr_B \otimes I_C) \circ (\Tr_A \otimes I_{BC}) = (\Tr_A \otimes I_C) \circ (\Tr_B \otimes I_{AC}).
\end{equation}
None of the maps involved in this identity are equal to each other.

The notion of polymorphic contraction was invented to simplify these identities. In order to do that, we should consider a partial function $\mathfrak{C}_a: \mathcal{D}_{\text{loc}} \to \mathcal{D}_{\text{loc}}$ which is defined as follows.
\begin{defi}
For a cell $a\in \mathcal{P}$
\begin{equation}
\mathfrak{C}_a(\rho^X) = I_{X\setminus a} \otimes \Tr_a(\rho^X)
\end{equation}
for $a\subset X$.
\end{defi}\footnote{Of course, $a$ does not necessarily have to be a cell, but this is all we need in this paper.}
Now, $\mathfrak{C}_a \circ \mathfrak{C}_b = \mathfrak{C}_b \circ \mathfrak{C}_a$ whenever the expression is well-defined. This makes the sense in which the partial trace operations ``commute'' more precise.

The purpose of the polymorphic extension is to make the notion of extending state precise. Let us start with the definition.
\begin{defi}
For a cell $a\in \mathcal{P}$, a cluster $A\in \mathcal{C}$, and $X\subset V$ such that $X\cap a = \emptyset$ and $\bar{A} \supset (\mathcal{N}(a)\cap X) \cup a$,
\begin{equation}
\mathfrak{E}_{a^A}(\rho^X) = I_{X\setminus \mathcal{N}(a)} \otimes [\mathcal{R}_{A}]_{\mathcal{N}(a)\cap X}^{(\mathcal{N}(a)\cap X)\cup a}(\rho^X).
\end{equation}
\end{defi}
Again $\mathfrak{E}_{a^A}$ is a partial function from $\mathcal{D}_{\text{loc}}$ to $\mathcal{D}_{\text{loc}}$.
Physically, a polymorphic extension sends a state supported on $X$ to a state supported on $X\cup a$ by applying a universal recovery map in the vicinity of $a$. The choice of the universal recovery map depends on $X$ as well as $a$, because the relevant subsystem from which the universal recovery map is defined($(\mathcal{N}(a)\cap X)\cup a$) depends on both of these sets. This is to be contrasted with the polymorphic contractions, where the relevant CPTP map is fixed to be the partial trace operation. By composing polymorphic extensions, one can gradually grow  the state until it is supported on $V$. Conversely, by composing polymorphic contractions, one can send the given state to another state in a smaller region. By defining the polymorphic extesnion, we no longer have to specify the full domain and the codomain of the universal recovery map. Once we specify the cell($a$) and the cluster that contains this cell($A$), the map is defined completely.

Now, we will develop a formalism to represent a sequence of polymorphic extensions and contractions. There are three types of data that we need to specify: the nature of the partial function, the choice of the cell, and the choice of the cluster. Obviously, these conventions greatly depend on the details of the Markovian marginal.

\subsection{Formalism for Theorem \ref{thm:1D}}
We develop a formalism to represent a variety of states that are defined in terms of the $\epsilon-$Markovian marginal described in Theorem \ref{thm:1D}. For each cell $[i]$, it is easy to verify that there are at most two clusters that contain the cell, e.g., $\{[i-1], [i] \}$ and $\{[i], [i+1] \}$;see FIG.\ref{fig:clusters_LR}. We shall refer to the first type as the left type and the second type as the right type. Once the choice of cell is clear from the context, we shall use $L$ and $R$ to refer to these clusters. We represent the polymorphic extensions and contractions with the following notation:
\begin{equation}
[i]^{I},
\end{equation}
where $[i]=\{2i-1,2i \}$ is one of the cells and $I=-1, L,$ or $R.$ If $I=-1$, the partial function is $\mathfrak{C}_{[i]}$. If $I=L$, the partial function is $\mathfrak{E}_{[i]^L}$. If $I=R,$ the partial function is $\mathfrak{E}_{[i]^{R}}$.
\begin{figure}[h]
\includegraphics[width=4in]{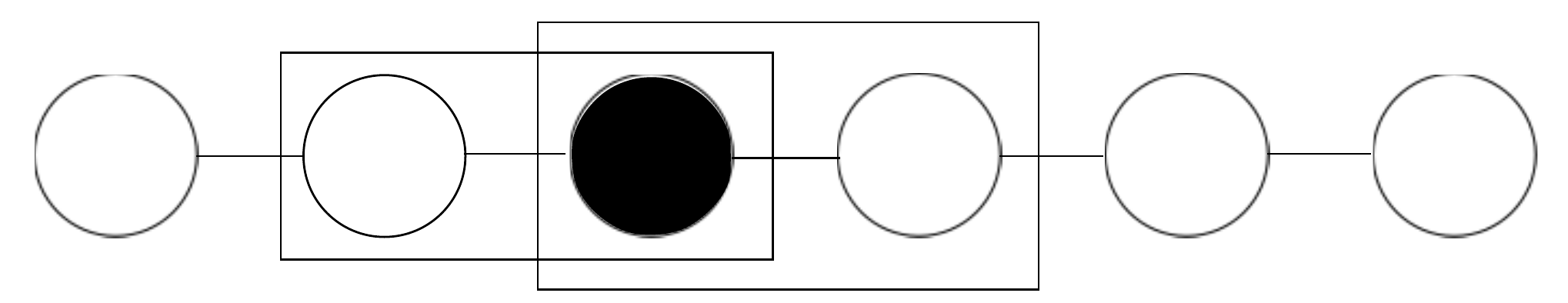}
\caption{Each cells(circles) either belong to the left or the right cluster. Once the cell is specified, the cluster can be completely specified by its position relative to the cell. For a cell $[i]$, its polymorphic extension defined in terms of the left cluster is $[i]^L$ and its polymorphic extension defined in terms of the right cluster is $[i]^R$.\label{fig:clusters_LR}}
\end{figure}

For specifying a state that is created by composing a sequence of such partial functions, we use strings with the following set of conventions. At the first step, the input to the partial function is assumed to be $1$, a scalar. In this case, the subsystem $B$ of the universal recovery map is an empty set. The resulting state is thus a reduced density matrix of the marginal that defines the universal recovery map, over a subsystem $C$. Since this rule applies to all the states we study, we do not explicitly specify the initial input. Next, we specify the sequence of partial functions by concatenating these symbols, starting from the left. The resulting string can be thought as an order in which the state is created. For example, $[1]^R[2]^L[3]^L$ represents a state over $[1]\cup [2] \cup [3]$ which is grown from $[1]$ to $[1]\cup [2]$, and then to $[1]\cup [2] \cup [3]$. Also, $[1]^R[2]^L[3]^L [2]^{-1}$ represents its reduced state over $[1]\cup [3]$.

\subsection{Formalism for Theorem \ref{thm:2D}}
Similar to the formalism developed for Theorem \ref{thm:1D}, we introduce a convention that specifies the cell, the cluster that includes the cell, and the choice of the partial function. We represent this data as follows:
\begin{equation}
[i,j]^I,
\end{equation}
where $[i,j]$ is the cell.(Recall that $i$ and $j$ represents the $x$ and the $y$ coordinate of the center of the hexagon that represents the cell.) If $I=-1$, it represents $\mathfrak{C}_{[i,j]}$. Otherwise this symbol represents a polymorphic extension whose underlying state is defined by a cluster. Since all the marginals on clusters consisting of three cells is completely determined by the marginals on clusters consisting of four cells, and because there are at most four such clusters that contain a given cell, we need a convention to specify these $4$ choices.
They are labeled by their relative position with respect to $[i,j]$. The relative position is specified in terms of four letters, $U, D, L,$ and $R$, each representing up, down, left, and right; see FIG.\ref{fig:clusters_UDLR}.

\begin{figure}[h]
\includegraphics[width=4in]{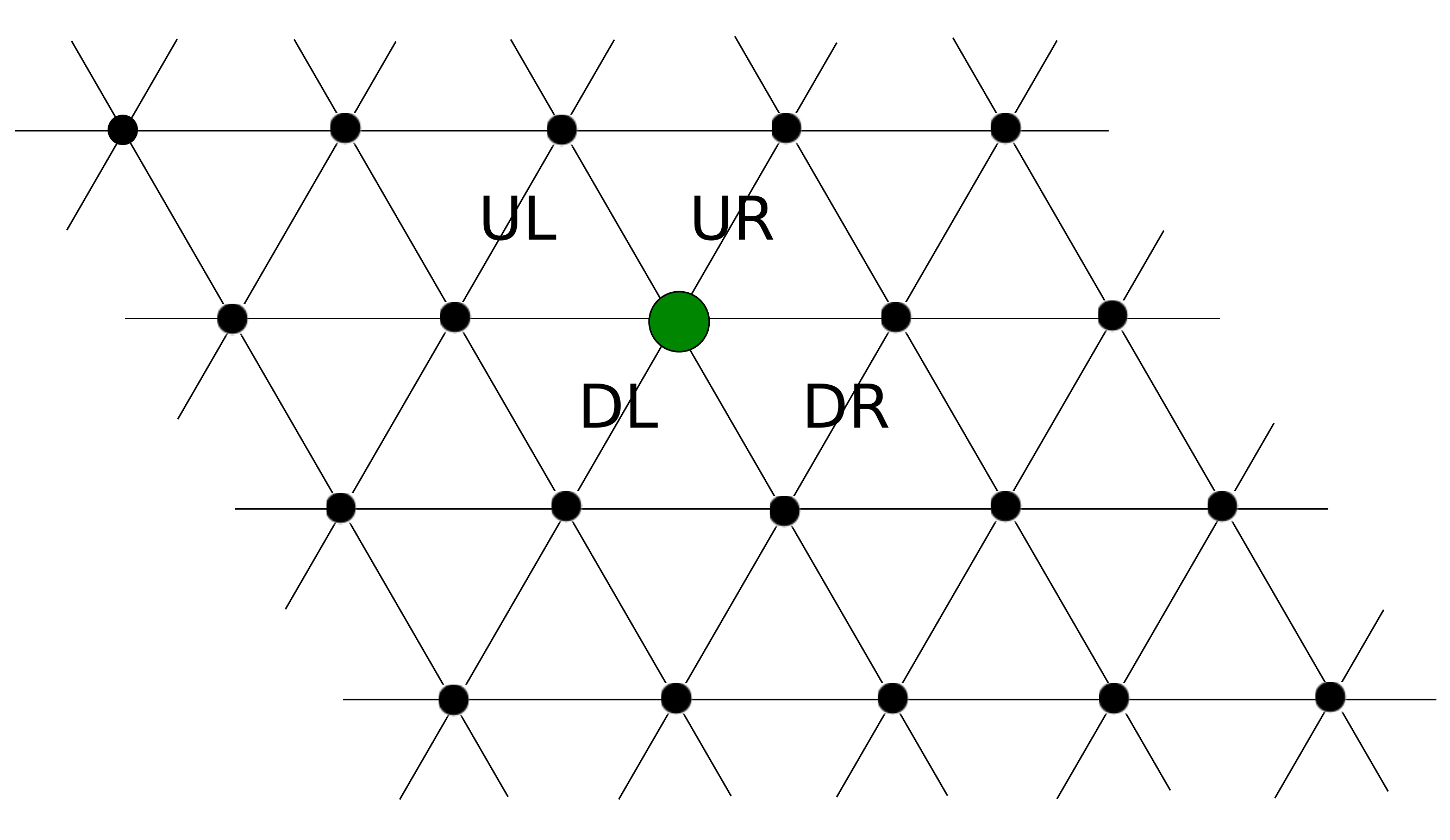}
\caption{A graph that represents the cells(circles) and their adjacency relations. For two cells $[i_1,j_1]$ and $[i_2,j_2]$, $[i_1,j_1]{-}[i_2,j_2]$ if and only if they are connected by an edge in this graph. For each cells, there are at most 4 possible clusters that contain the cell. Each of these clusters are either on the up-right, up-left, down-right, or down-left relative to the cell. The green circle is chosen as an example. Once a cell is specified, these clusters are referred to $UR$, $UL$, $DR$, and $DL$.\label{fig:clusters_UDLR}}
\end{figure}

Again, we suppress the composition symbol and specify the order in which the partial functions are applied. The leftmost symbol represents the first partial function and the rightmost symbol represents the last partial function. The initial input is always assumed to be a scalar $1$.

\subsection{Manifest relations}
The polymorphic extensions and contractions obey a certain set of relations. These relations are manifest in a sense that no extra assumption is necessary to ensure their validity. First, we note that $SS'' = S'S''$ for all $S=S'$ and $S''$, provided that the strings on both sides are well-defined. Furthermore, this relation is robust. If $S=S'$, it implies that the state represented by these strings are supported on the same space. Then the sequence of CPTP maps that are applied by the string $S''$ should be the same on both sides. Since these maps are norm-nonincreasing, $\|S-S' \|_1\leq O(\epsilon)$ implies
\begin{equation}
\| SS'' - S'S''\|_1\leq O(\epsilon).
\end{equation}
We shall use a short-hand notation of $S\aeq S'$ and $SS'' \aeq S'S''$ to denote these facts. The following fact, which is an exact identity, follows straightforwardly from the definition.
\begin{prop}
For all $\rho^X$ such that $X \supset a\cup b$, $\mathfrak{C}_a \circ \mathfrak{C}_b(\rho^X) = \mathfrak{C}_a \circ \mathfrak{C}_b(\rho^X)$.
\end{prop}

Similar relations hold between two polymorphic extensions, and also between a polymorphic extension and a polymorphic contraction, provided that the relevant CPTP maps are supported on disjoint subsystems.
\begin{prop}
If $(a\cup \mathcal{N}(a)) \cap (b\cup \mathcal{N}(b))=\emptyset$, $X\cap (a\cup b) = \emptyset$, $\mathcal{N}(a)\cap X \subset A$, and $\mathcal{N}(b)\cap X \subset B$, $\forall \rho^X$ $\mathfrak{E}_{a^A} \circ \mathfrak{E}_{b^B}(\rho^X) = \mathfrak{E}_{b^B} \circ \mathfrak{E}_{a^A}(\rho^X).$
\end{prop}
\begin{proof}
First note that the required conditions ensure that both partial functions are well-defined.
\begin{align}
\mathfrak{E}_{a^A} \circ \mathfrak{E}_{b^B}(\rho^X) &= \mathfrak{E}_{a^A}(I_{X\setminus \mathcal{N}(b)} \otimes [\mathcal{R}_B]_{\mathcal{N}(b)\cap X}^{(\mathcal{N}(b)\cap X)\cup b}(\rho^X)) \\
&=(I_{(X\cup b)\setminus \mathcal{N}(a)} \otimes [\mathcal{R}_A]_{\mathcal{N}(a)\cap (X\cup b)}^{(\mathcal{N}(a)\cap(X\cup b))\cup a})\circ (I_{X\setminus \mathcal{N}(b)} \otimes [\mathcal{R}_B]_{\mathcal{N}(b)\cap X}^{(\mathcal{N}(b)\cap X)\cup b}(\rho^X) \\
&=(I_{(X\cup b)\setminus \mathcal{N}(a)} \otimes [\mathcal{R}_A]_{\mathcal{N}(a)\cap X}^{(\mathcal{N}(a)\cap X)\cup a})\circ (I_{X\setminus \mathcal{N}(b)} \otimes [\mathcal{R}_B]_{\mathcal{N}(b)\cap X}^{(\mathcal{N}(b)\cap X)\cup b}(\rho^X).
\end{align}
Since the two universal recovery maps in the last line are supported on two disjoint subsystems, one can exchange the order in which these maps are applied, modulo an appropriate change of the identity superoperators. By applying the same sequence of logic on $\mathfrak{E}_{b^B} \circ\mathfrak{E}_{b^B}(\rho^X)$, one can show that the two expressions are equal to each other.
\end{proof}
Similarly,
\begin{prop}
If $a\cap X = \emptyset$, $b\subset X$, $\mathcal{N}(a)\cap X \subset A$, and $\mathcal{N}(a)\cap b=\emptyset$, $\mathfrak{C}_b\circ \mathfrak{E}_{a^A}(\rho^X) = \mathfrak{E}_{a^A}\circ \mathfrak{C}_b(\rho^X)$.
\end{prop}

The Markovian marginals introduced in Theorem \ref{thm:1D} and Theorem \ref{thm:2D} possess a useful property: that the neighbors of two cells overlap with each other if and only if the two cells are adjancet to each other. More formally, for the Markovian marginal in Theorem \ref{thm:1D} $(\mathcal{N}([i])\cup [i])\cap (\mathcal{N}([j])\cup [j])=\emptyset$ if and only if $[i] \centernot{-} [j]$. This implies that any two partial functions commute whenever the relevant cells are not adjacent to each other. These relations are summarized in Table \ref{table:1D}.
\begin{table}[h]
\begin{tabular}{c c}
\hline
Condition & Relations \\ \hline
Well-defined string and $S\aeq S'$& $SS'' \aeq S'S''$ \\ \hline
Well-defined string & $S[i]^{-1}[j]^{-1}S' = S[j]^{-1}[i]^{-1}S'$ \\ \hline
Well-defined string and $[i]\centernot{-}[j]$ & $S[i]^I[j]^JS' = S[j]^J[i]^IS'$ \\
\hline
\end{tabular}
\caption{Manifest relations for Theorem \ref{thm:1D}.Here $S$, $S'$, and $S''$ can be arbitrary strings, provided that the strings on both sides of the relations are well-defined. The superscripts($I$ and $J$) take the value in $-1,L,$ and $R.$ \label{table:1D}}
\end{table}

A similar conclusion holds for the Markovian marginal in Theorem \ref{thm:2D}. That is,
$((\mathcal{N}([i_1,j_1]) \cup [i_1,j_1]) \cap(\mathcal{N}([i_2,j_2]) \cup [i_2,j_2]) )=\emptyset$ if and only if $[i_1,j_1] \centernot{-} [i_2,j_2]$. These relations are summarized in Table \ref{table:2D}.
\begin{table}[h]
\begin{tabular}{c c }
\hline
Condition & Relations \\ \hline
Well-defined string and $S\aeq S'$ &  $SS'' \aeq S'S''$ \\ \hline
Well-defined string & $S[i_1,j_1]^{-1}[i_2,j_2]^{-1}S' = S[i_2,j_2]^{-1}[i_1,j_1]^{-1}S'$ \\ \hline
Well-defined string and $[i_1,j_1]\centernot{-}[i_2,j_2]$ & $S[i_1,j_1]^I[i_2,j_2]^JS' = S[i_2,j_2]^J[i_1,j_1]^IS'$ \\
\hline
\end{tabular}
\caption{Manifest relations for Theorem \ref{thm:2D}. Here $S$ and $S'$ can be arbitrary strings, provided that the strings on both sides of the relations are well-defined. The superscripts($I$ and $J$) take the value in $-1, UR, UL, DR,$ and $DL.$\label{table:2D}}
\end{table}

\section{Proof of Theorem \ref{thm:1D}\label{section:Proof1}}
So far, we have not yet used the local consistency and the local Markov condition. These conditions lead to what we call as \emph{derived relations.} By using these relations together with the manifest relations(Table \ref{table:1D}),  we can prove our main result. The following convention will be useful:
\begin{defi}
For $m'\geq m$,
\begin{equation}
\Pi_{i=m}^{m'} S_i = (\Pi_{i=m}^{m'-1} S_i) S_{m'}
\end{equation}
where $\Pi_{i=m}^j S_i$ is set to be an empty string if $j<m$.
\end{defi}

Our main technical statement is that the reduced density matrix of
\begin{equation}
[1]^{R}(\Pi_{i=1}^{\frac{n}{2}-1}[i+1]^L)   \label{eq:proposed_state_1D}
\end{equation}
over $[i]\cup [i+1]$ is equal to $\rho^{[i]\cup [i+1]}\in \mathcal{M}$ up to a trace distance that is bounded by $O(n\epsilon)$ for all $i=1, \cdots, \frac{n}{2}-1$. The proof of this statement can be broken down into three steps. For $1\leq m \leq \frac{n}{2}-1$  we first show that
\begin{equation}
[1]^R(\Pi_{i=1}^{\frac{n}{2}-1}[i+1]^{L})(\Pi_{i=1}^{m-1}[i]^{-1}) \aeqm [m]^R (\Pi_{i=1}^{\frac{n}{2}-m}[m+i]^L).
\end{equation}
Second, we show that
\begin{equation}
[m]^{R} (\Pi_{i=1}^{\frac{n}{2}-m}[m+i]^L) \underset{O((\frac{n}{2}-m)\epsilon)}{\approx} [n]^L (\Pi_{i=1}^{\frac{n}{2}-m} [\frac{n}{2}-i]^R).
\end{equation}
 Third, we show that
\begin{equation}
[n]^{L} (\Pi_{i=1}^{\frac{n}{2}-m} [\frac{n}{2}-i]^R) (\Pi_{i=0}^{\frac{n}{2}-m-2} [\frac{n}{2}-i]^{-1}) \underset{O(m\epsilon)}{\approx} [m]^R[m+1]^L.
\end{equation}
By invoking the triangle inequality for the trace norm, this leads to a conclusion that the marginal of the proposed state over $[m]\cup [m+1]$ is close to the state $[m]^R[m+1]^L$ up to a trace distance bounded by $O(n\epsilon)$. By the Markov condition, this state is close to the $\rho^{[m]\cup [m+1]} \in \mathcal{M}$ up to an $O(\epsilon)$ deviation. This concludes the proof.  Each of these steps can be further divided into a set of elementary derived relations, which we summarize in Table. \ref{table:derived_1D}
\begin{table}[h]
\begin{tabular}{c| c c c}
Name & Relations & Derivation & Appearance \\ \hline
$\cdot$&$[i]^L S\aeq[i]^R S$ & Consistency & $\cdot$ \\ \hline
Forward contraction&$[i]^R[i+1]^L[i]^{-1}S \aeq [i+1]^LS$ & Markov &Lemma \ref{lemma:1D_forward_contraction}\\ \hline
Cell exchange&$[i]^R[i+1]^LS \aeq [i+1]^L[i]^RS$ & Markov & Lemma  \ref{lemma:1D_cell_exchange}\\ \hline
Backward contraction&$[i]^L[i-1]^{R}[i]^{-1}S \aeq [i-1]^RS$ & Markov & Lemma \ref{lemma:1D_backward_contraction} \\ \hline
\end{tabular}
\caption{Here $S$ can be an arbitrary string, provided that the strings on both sides are well-defined. The derivation lists a set of invoked assumptions, consistency referring to the local consistency condition and the Markov referring to the local Markov condition.\label{table:derived_1D}}
\end{table}

The forward contraction can be achieved by using the following lemma as a subroutine.
\begin{lem}\label{lemma:1D_forward_contraction}
(Forward contraction) For $1\leq i\leq \frac{n}{2}-1$
\begin{equation}
[i]^R[i+1]^L[i]^{-1} \aeq [i+1]^L.
\end{equation}
\end{lem}
\begin{proof}
\begin{equation}
[i]^R[i+1]^L[i]^{-1} \aeq I_{[i+1]} \otimes \Tr_{[i]} (\rho^{[i]\cup [i+1]})
\end{equation}
by the local Markov condition, where $\rho^{[i]\cup [i+1]} \in \mathcal{M}$.
\end{proof}
Since $[i+1]^L\aeq [i+1]^R$ for $1\leq i\leq n-2$, for those $i$ we conclude that $[i]^R[i+1]^L[i]^{-1} \aeq [i+1]^R$.\footnote{Otherwise, either $[i+1]^L$ or $[i+1]^R$ is undefined.} By repeatedly applying these relations $O(m)$ times, we conclude that$[1]^R(\Pi_{i=1}^{n-1}[i+1]^{L})(\Pi_{i=1}^{m-1}[i]^{-1}) \aeqm [m]^R (\Pi_{i=1}^{n-m}[m+i]^L).$

The next step is to flip the order of the cells. This can be achieved by using the following lemma as a subroutine.
\begin{lem}\label{lemma:1D_cell_exchange}
(Cell exchange) For $1\leq i \leq \frac{n}{2}-1$
\begin{equation}
[i]^{R}[i+1]^L \aeq [i+1]^L[i]^R
\end{equation}
\end{lem}
\begin{proof}
Note that
\begin{equation}
\begin{aligned}
\rho^{[i]\cup [i+1]} &\aeq [i]^R[i+1]^L \\
\rho^{[i]\cup [i+1]} &\aeq [i+1]^L[i]^R
\end{aligned}
\end{equation}
by the local Markov condition. By using the triangle inequality for the trace norm, we conclude  $[i]^{R}[i+1]^L \aeq [i+1]^L[i]^R.$
\end{proof}
Since $[i+1]^L\aeq [i+1]^R$ for $1\leq i\leq n-2$, for those $i$ we conclude that $[i]^R[i+1]^L\aeq [i+1]^R[i]^R$. After applying these operations, we use the manifest relations to move $[i]^R$ to the right end of the string. By applying the same set of relations for such $i$ and then applying the relation $[\frac{n}{2}-1]^R[\frac{n}{2}]^L \aeq [\frac{n}{2}]^L[\frac{n}{2}-1]^R$ at the end, we conclude that $[m]^{R} (\Pi_{i=1}^{\frac{n}{2}-m}[m+i]^L) \underset{O((\frac{n}{2}-m)\epsilon)}{\approx} [\frac{n}{2}]^L (\Pi_{i=1}^{\frac{n}{2}-m} [\frac{n}{2}-i]^R).$

Now we are left with the last step, backward contraction. We use the following lemma as a subroutine.
\begin{lem}\label{lemma:1D_backward_contraction}
(Backward contraction) For $3\leq i \leq \frac{n}{2}$
\begin{equation}
[i]^L[i-1]^R[i]^{-1} \aeq [i-1]^R.
\end{equation}
\end{lem}
\begin{proof}
By the local Markov condition,
\begin{equation}
[i]^L[i-1]^R[i]^{-1} \aeq I_{[i-1]} \otimes \Tr_{[i]}(\rho^{[i-1] \cup [i]}),
\end{equation}
where $\rho^{[i-1] \cup [i]}\in  \mathcal{M}$.
\end{proof}
Recall that $[i-1]^{R} \aeq [i-1]^L$ for $3\leq i\leq \frac{n}{2}$. For such $i$, we have $[i]^L[i-1]^R[i]^{-1} \aeq [i-1]^L.$ We can apply these relations recursively until we are left with $[m+1]^L[m]^R$. By the local Markov condition, the trace distance between this state and $\rho^{[m]\cup[m+1]}$ is bounded by $O(\epsilon)$. The trace distance between $\rho^{[m]\cup [m+1]}$ and the reduced state of the proposed state(Eq.\ref{eq:proposed_state_1D}) is bounded from above by $O(m\epsilon) + O((\frac{n}{2}-m)\epsilon) \leq O(n\epsilon)$ for all $1\leq m\leq n-1$. This completes the proof of Theorem \ref{thm:1D}.

\section{Proof of Theorem \ref{thm:2D}\label{section:Proof2}}
The proof of Theorem \ref{thm:2D} is analogous to that of Theorem \ref{thm:1D}. Let us first sketch an overview. We  propose the following specific state:
\begin{equation}
[:,1]^U(\Pi_{i=1}^{n-1} [:,i+1]^{D}),\label{eq:proposed_state_2D}
\end{equation}
where
\begin{equation}
[:,i]^U = [1,i]^{UR}\Pi_{j=1}^{n-1}[j+1,i]^{UL}
\end{equation}
and
\begin{equation}
[:,i]^{D} = [1,i]^{DR} \Pi_{j=1}^{n-1}[j+1,i]^{DL},
\end{equation}
and then prove that the trace distance between the reduced density matrix of this state over the given clusters and the marginals in $\mathcal{M}$ are bounded from above by $O(n^2\epsilon)$ for all the clusters.

We would like to point out a similarity between Eq.\ref{eq:proposed_state_1D} and Eq.\ref{eq:proposed_state_2D}. In Eq.\ref{eq:proposed_state_1D} the state is created sequentially from the left to the right. At each steps, the polymorphic extensions are supported on at most two cells. Similarly, in Eq.\ref{eq:proposed_state_2D} the state is created sequentially from the bottom to the top. At each steps, the polymorphic extensions are supported on at most two rows of cells.  As we shall see, we can apply an analogue of the forward contraction, cell exchange, and backward contraction to our proposed state, i.e., Eq.\ref{eq:proposed_state_2D}. Instead of removing a cell at each steps, we remove a row of cells. Also, instead of exchanging the order of the cells at each steps, we exchange the order of the rows at each steps.\footnote{Actually, the order of the cells within the row changes, but we will be able to deal with this subtlety.}

These procedures are explained in Section \ref{section:row_relations}. At the end of applying these procedures, we arrive at a conclusion:
\begin{equation}
[:,1]^U(\Pi_{i=1}^{n-1} [:,i+1]^{D}) (\Pi_{i=1}^{m-1} [:,m]^{-1}) (\Pi_{i=m+2}^{n}[:,i]^{-1}) \underset{O(n^2\epsilon)}{\approx} [:,m]^{U}[:,m+1]^D. \label{eq:2D_middle_step}
\end{equation}
That is, the reduced state of the proposed state over two rows of cells is approximately equal to a certain state that is specified by a string of cells(as well as the clusters) that are confined in just two rows. Then we derive analogues of forward contraction, cell exchange, and backward contraction for these two rows. This is the content of Section \ref{section:supercell_relations}. By applying these operations, we conclude that the reduced density matrices of Eq.\ref{eq:proposed_state_2D} over the clusters is approximately equal to certain states that are specified by strings of bounded length. At this point, we use the local Markov condition and show that this state is close to the marginals in $\mathcal{M}$ that are supported on the same set of cells, thus completing the proof.

We begin by deriving a set of relations that involve bounded number of cells in Section \ref{section:localized_relations}. These relations are combined together derive the relations involving rows of cells and the relations involving two rows.

\subsection{Localized relations\label{section:localized_relations}}
We begin by deriving relations that involve bounded number of cells.  For the Markovian marginal discussed in Theorem \ref{thm:2D}, there are certain relations that are not entirely obvious from the given local Markov condition. These relations are \emph{inherited}, in a sense that it follows from a local Markov condition that is inherited from the given local Markov condition. That is, the relevant local Markov condition is not specified explicitly in Theorem \ref{thm:2D}, but it nevertheless can be shown to follow from the given local Markov conditions. As these relations are used frequently, it is best to derive them first.

All the inherited relations that we derive here concern two cells that are adjacent to each other, e.g., pairs such as $[i_1,j_1]$ and $[i_1+1,j_1]$ or pairs such as $[i_1,j_1]$ and $[i_1,j_1+1]$. The relations will be determined by a polymorphic extension that is defined in terms of the marginal over a cluster that contains both of the cells. Notice that there is an ambiguity in this statement. Specifically, given two such cells, there are four distinct clusters that contain these cells, two of which have $3$ cells and two of which have $4$ cells. The clusters with $3$ cells are contained within one of the clusters with $4$ cells, and there is no ambiguity between clusters of different sizes. The ambiguity lies on clusters of the same size; see FIG.\ref{fig:inheritance}. We need to treat these two cases separately.
\begin{figure}[h]
\subfigure[]{\includegraphics[width=2in]{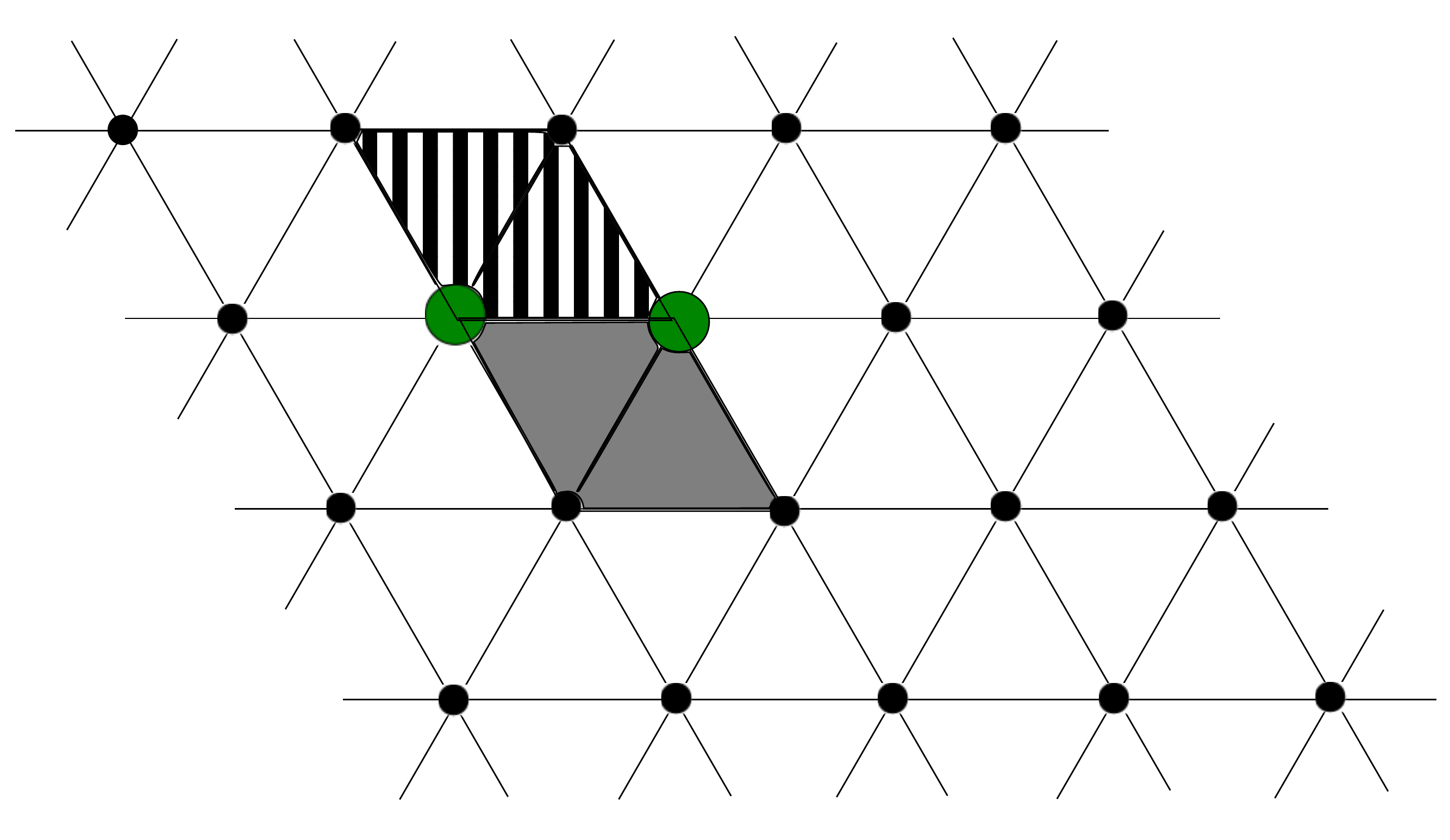}}
\subfigure[]{\includegraphics[width=2in]{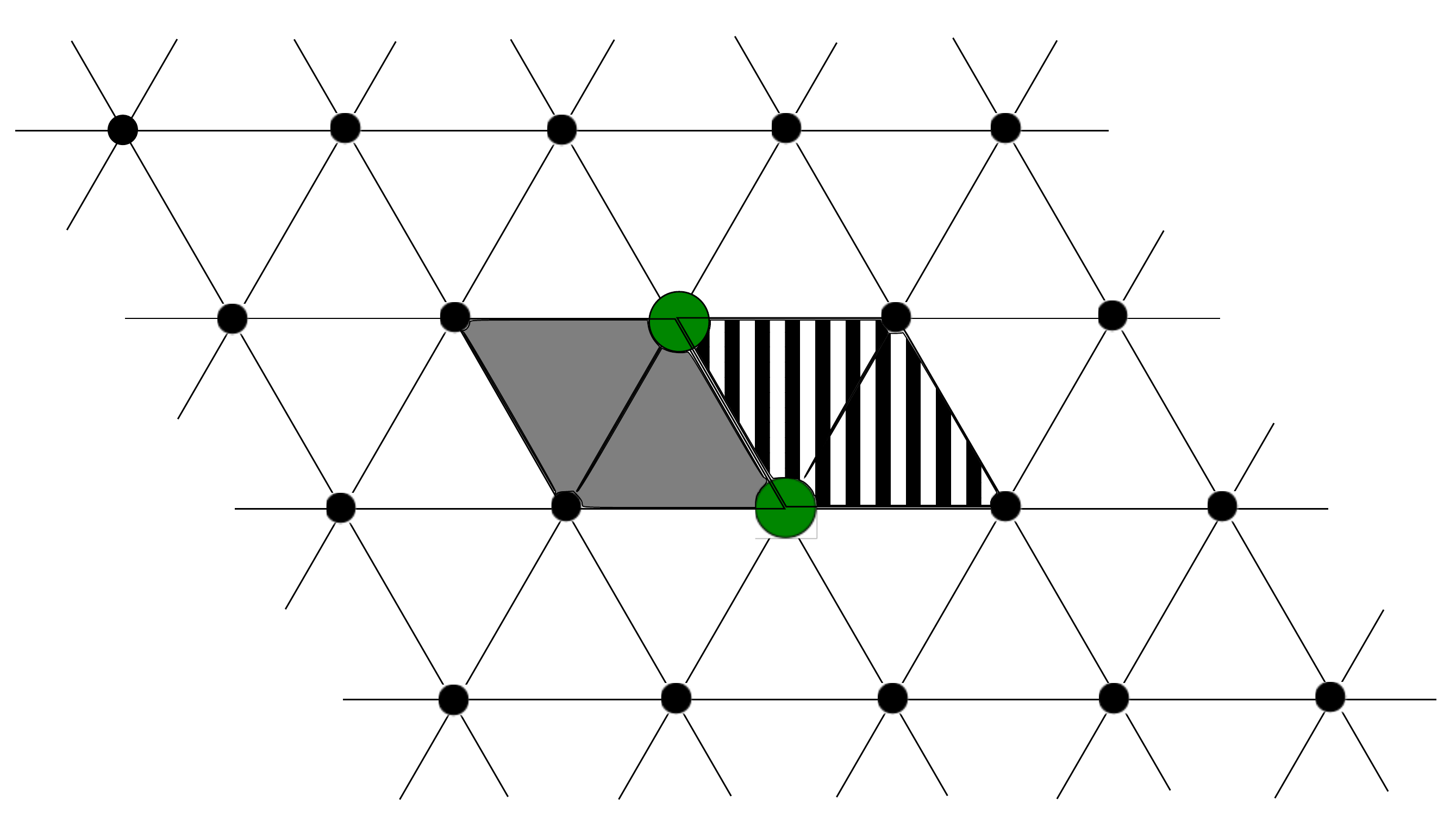}}
\caption{For certain pairs of cells that are adjacent to each other(green circles), there are two possible clusters that contain both of the cells. One is striped, and the other one is gray in this figure. Since the marginals over these clusters are only assumed to be approximately locally consistent, we should treat their reduced density matrices over the two cells separately.\label{fig:inheritance}}
\end{figure}

\begin{lem}
For $i=1,\cdots, n-1$ and $j=2,\cdots, n-1$
\begin{equation}
[i,j]^{UR}[i+1,j]^{UL} \aeq [i+1,j]^{DL} [i,j]^{DR}.
\end{equation}
\label{lemma:2D_inheritance}
\end{lem}
\begin{proof}
Note that the cluster $UR$ relative to $[i,j]$ and the cluster $UL$ relative to $[i+1,j]$ are the same cluster. Let us denote this cluster as $A$ and the marginal over this cluster as $\rho^{\bar{A}}\in \mathcal{M}$. Similarly, the cluster $DL$ relative to $[i+1,j]$ and the cluster $DR$ relative to $[i,j]$ are the same cluster. Let us denote this cluster as $B$ and the given marginal over this cluster as $\rho^{\bar{B}} \in  \mathcal{M}$.

By the local Markov condition that follow from $\mathcal{C}_3$(FIG.\ref{fig:CMI2D_3-2}) we have
\begin{equation}
I([i+1,j]: ([i,j]\cup [i,j+1]) \setminus \mathcal{N}([i+1,j])]| \mathcal{N}([i+1,j]) \cap \bar{A})_{\rho^{\bar{A}}} \leq \epsilon^2.
\end{equation}
By the strong subadditivity of entropy,\cite{Lieb1973}
\begin{equation}
\begin{aligned}
I([i+1,j] : [i,j] \setminus \mathcal{N}([i+1,j])]| \mathcal{N}([i+1,j] \cap \bar{A}) )_{\rho^{\bar{A}}}  \\ \leq I([i+1,j]: ([i,j]\cup [i,j+1]) \setminus \mathcal{N}([i+1,j])]| \mathcal{N}([i+1,j]) \cap \bar{A} )_{\rho^{\bar{A}}},
\end{aligned}
\end{equation}
because $[i,j] \setminus \mathcal{N}([i+1,j])] \subset ([i,j]\cup [i,j+1]) \setminus \mathcal{N}([i+1,j])]$. By applying Theorem \ref{thm:universal_recovery} to this expression, we conclude that the trace distance between $[i,j]^{UR}[i+1,j]^{UL}$ and the reduced density matrix of $\rho^{\bar{A}}$ over the two cells  is bounded from above by $O(\epsilon)$. Applying the same line of logic, we conclude that the trace distance between $[i+1,j]^{DL} [i,j]^{DR}$ and the reduced density matrix of $\rho^{\bar{B}}$ over the two cells is bounded from above by $O(\epsilon).$ The trace distance between the reduced density matrix of $\rho^{\bar{A}}$ and the reduced density matrix of $\rho^{\bar{B}}$ over the two cells are bounded from above by $O(\epsilon)$ due to the local consistency condition. By applying the triangle inequality to these three trace distance bounds, we conclude $[i,j]^{UR}[i+1,j]^{UL} \aeq [i+1,j]^{DL} [i,j]^{DR}.$
\end{proof}

With an appropriate choice of clusters, one can also prove the following lemma with the same logic.
\begin{lem}
\label{lemma:2D_inheritance_vertical}
For $i=2,\cdots, n-1$ and $j=1,\cdots, n-1$,
\begin{equation}
[i,j]^{UR} [i,j+1]^{DR} \aeq [i,j+1]^{DL}[i,j]^{UL}
\end{equation}
\end{lem}

There are relations that involve $4$ cells.
\begin{lem}
\label{lemma:2D_4cell_permutation}
For $i,j= 1, \cdots, n-1$,
\begin{equation}
[i,j]^{UR}[i+1,j]^{UL} [i,j+1]^{DR}[i+1,j+1]^{DL} \aeq [i+1,j+1]^{DL} [i,j+1]^{DR}[i+1,j]^{UL}[i,j]^{UR}.
\end{equation}
\end{lem}
\begin{proof}
Note that all the specified clusters are in fact the same cluster. Let us refer to this cluster as $A$. By the inherited local Markov conditions, the reduced density matrix of the marginal on this cluster over $[i,j]\cup [i+1,j]$ is close to $[i,j]^{UR}[i+1,j]^{UL}$ up to a trace distance $O(\epsilon)$. By a local Markov condition  that follows from $\mathcal{C}_3$(FIG.\ref{fig:CMI2D_3-1}) the reduced density matrix over $[i,j] \cup [i+1,j] \cup [i,j+1]$ is close to
$[\mathcal{R}_{[i,j+1]^A}](\rho^{[i,j] \cup [i+1,j]})$, where $\rho^{[i,j] \cup [i+1,j]}$ is the reduced density matrix of $\rho^{\bar{A}}$ over $[i,j]\cup [i+1,j]$. By using the triangle inequality, we conclude that the reduced density matrix over $[i,j]\cup [i+1,j] \cup [i,j+1]$ is close to $[i,j]^{UR}[i+1,j]^{UL}[i,j+1]^{DR}$ up to a trace distance $O(\epsilon)$. We then invoke a local Markov condition that follow from $\mathcal{C}_4$(FIG.\ref{fig:CMI2D_4-2}) to conclude that $\rho^{\bar{A}}$ is close to $[i,j]^{UR}[i+1,j]^{UL} [i,j+1]^{DR}[i+1,j+1]^{DL}$ up to a trace distance $O(\epsilon)$. A similar line of logic can be applied to  $[i+1,j+1]^{DL} [i,j+1]^{DR}[i+1,j]^{UL}[i,j]^{UR},$ by invoking the local Markov conditions in FIG.\ref{fig:CMI2D_3-5} and FIG.\ref{fig:CMI2D_4-4}. We then invoke the triangle inequality to conclude $[i,j]^{UR}[i+1,j]^{UL} [i,j+1]^{DR}[i+1,j+1]^{DL} \aeq [i+1,j+1]^{DL} [i,j+1]^{DR}[i+1,j]^{UL}[i,j]^{UR}.$
\end{proof}

\begin{lem}
\label{lemma:2D_contraction_LD}
For $i,j=1, \cdots, n-1$,
\begin{equation}
[i,j]^{UR}[i+1,j]^{UL} [i,j+1]^{DR}[i+1,j+1]^{DL}[i,j]^{-1} \aeq [i+1,j+1]^{DL} [i+1,j]^{UL}[i,j+1]^{DR}
\end{equation}
\end{lem}
\begin{proof}
Note that all the specified clusters are in fact the same cluster. We refer to this cluster as $A$. By the inherited local Markov conditions, the reduced density matrix of the marginal on this cluster over $[i,j]\cup [i+1,j]$ is close to $[i,j]^{UR}[i+1,j]^{UL}$ up to a trace distance $O(\epsilon)$. By a local Markov condition  that follow from $\mathcal{C}_3$(FIG.\ref{fig:CMI2D_3-1}) the reduced density matrix over $[i,j] \cup [i+1,j] \cup [i,j+1]$ is close to
$[\mathcal{R}_{[i,j+1]^A}](\rho^{[i,j] \cup [i+1,j]})$, where $\rho^{[i,j] \cup [i+1,j]}$ is the reduced density matrix of $\rho^{\bar{A}}$ over $[i,j]\cup [i+1,j]$. By using the triangle inequality, we conclude that the reduced density matrix over $[i,j]\cup [i+1,j] \cup [i,j+1]$ is close to $[i,j]^{UR}[i+1,j]^{UL}[i,j+1]^{DR}$ up to a trace distance $O(\epsilon)$. We then invoke a local Markov condition that follow from $\mathcal{C}_4$(FIG.\ref{fig:CMI2D_4-2}) to conclude that $\rho^{\bar{A}}$ is close to $[i,j]^{UR}[i+1,j]^{UL} [i,j+1]^{DR}[i+1,j+1]^{DL}$ up to a trace distance $O(\epsilon)$. The reduced density matrix of $[i,j]^{UR}[i+1,j]^{UL} [i,j+1]^{DR}[i+1,j+1]^{DL}$ is thus close to the reduced density matrix of $\rho^{\bar{A}}$ over $[i+1,j]\cup[i+1,j+1] \cup [i,j+1]$. Now, we can apply a similar line of logic to show that the reduced density matrix of $\rho^{\bar{A}}$ over $[i+1,j]\cup[i+1,j+1] \cup [i,j+1]$ to show that it is close to $[i+1,j+1]^{DL}[i+1,j]^{UL} [i,j+1]^{DR}$ up to a trace distance $O(\epsilon)$. Specifically, we use the inherited local Markov condition to show that the reduced density matrix of $\rho^{\bar{A}}$ over $[i+1,j]\cup [i+1,j+1]$ is close to $[i+1,j+1]^{DL}[i+1,j]^{UL}$. Then we use the local Markov condition in $\mathcal{C}_3$(FIG.\ref{fig:CMI2D_3-4}). By applying the triangle inequality, the claim is proved.
\end{proof}

\begin{lem}
\label{lemma:2D_contraction_UR}
For $i,j=1, \cdots, n-1$,
\begin{equation}
[i+1,j+1]^{DL}[i,j+1]^{DR}[i+1,j]^{UL}[i,j]^{UR}[i+1,j+1]^{-1} \aeq [i,j]^{UR}[i,j+1]^{DR}[i+1,j]^{UL}
\end{equation}
\end{lem}
\begin{proof}
The proof is essentially identical to that of Lemma \ref{lemma:2D_contraction_LD}, so we only provide a sketch. Note that the specified clusters are in fact the same cluster, and that the marginal over this cluster is close to $[i+1,j+1]^{DL}[i,j+1]^{DR}[i+1,j]^{UL}[i,j]^{UR}$ up to a trace distance of $O(\epsilon)$. This can be shown by invoking the inherited local Markov condition and the local Markov conditions corresponding to FIG.\ref{fig:CMI2D_3-5} and FIG.\ref{fig:CMI2D_4-4}. Then we sequentially use the local Markov conditions to show that the reduced density matrix of this marginal on $[i,j]\cup [i,j+1] \cup [i+1,j]$ is close to $[i,j]^{UR}[i,j+1]^{DR}[i+1,j]^{UL}$ up to a trace distance of $O(\epsilon)$, by invoking the local Markov condition in $\mathcal{C}_3$(FIG.\ref{fig:CMI2D_3-2}). We then apply the triangle inequality to prove the claim.
\end{proof}

These findings are summarzied in Table \ref{table:inheritance}.
\begin{table}[h]
\begin{tabular}{c c c}
Relations & Derivation & Appearance \\ \hline
$[i,j]^{UR} [i+1,j]^{UL}S \aeq [i+1,j]^{DL}[i,j]^{DR}S$ & Inherited from $\mathcal{C}_3$, Consistency & Lemma \ref{lemma:2D_inheritance} \\ \hline
$[i,j]^{UR} [i,j+1]^{DR}S \aeq [i,j+1]^{DL}[i,j]^{UL}S$ & Inherited from $\mathcal{C}_3$, Consistency & Lemma \ref{lemma:2D_inheritance_vertical} \\ \hline
$[i,j]^{UR}[i+1,j]^{UL} [i,j+1]^{DR}[i+1,j+1]^{DL}S$ &  Inherited from $\mathcal{C}_3$,  Markov, &  \\
 $\aeq [i+1,j+1]^{DL} [i,j+1]^{DR}[i+1,j]^{UL}[i,j]^{UR}S$ & FIG.\ref{fig:CMI2D_3-1}, \ref{fig:CMI2D_4-2}, \ref{fig:CMI2D_3-5}, \ref{fig:CMI2D_4-4} &Lemma \ref{lemma:2D_4cell_permutation}\\
\hline
$[i,j]^{UR}[i+1,j]^{UL} [i,j+1]^{DR}[i+1,j+1]^{DL}[i,j]^{-1}S$  & Inherited from $\mathcal{C}_3$,  Markov, &  \\
 $\aeq [i+1,j+1]^{DL} [i+1,j]^{UL} [i,j+1]^{DR}S$ & FIG.\ref{fig:CMI2D_3-1}, \ref{fig:CMI2D_4-2}, \ref{fig:CMI2D_3-4} & Lemma \ref{lemma:2D_contraction_LD} \\
\hline
$[i+1,j+1]^{DL}[i,j+1]^{DR}[i+1,j]^{UL}[i,j]^{UR}[i+1,j+1]^{-1}S$   & Inherited from $\mathcal{C}_3$, Markov, &  \\
$\aeq [i,j]^{UR}[i,j+1]^{DR}[i+1,j]^{UL}S$ & FIG.\ref{fig:CMI2D_3-5}, \ref{fig:CMI2D_4-4}, \ref{fig:CMI2D_3-2}  &  Lemma \ref{lemma:2D_contraction_UR} \\
\hline
\end{tabular}
\caption{Relations that involve bounded number of cells. These relations make use of the inherited local Markov conditions that are defined in terms of the clusters in $\mathcal{C}_3$. The string $S$ can be arbitrary, provided that the state represented by the string is well-defined on both sides of the relations.\label{table:inheritance}}
\end{table}

\subsection{Row relations\label{section:row_relations}}
By making use of the relations summarized in Table \ref{table:inheritance}, we build up analogues of the forward contraction, cell exchange, and backward contraction operations that were discussed in Section \ref{section:Proof1}. The main difference is that these operations are defined on rows, as opposed to the individual cells. Recall that in Section \ref{section:Proof1} we were able to reduce a string of length $O(n)$ to a string of length $2$. Our temporary goal is to reduce a string that involves $O(n)$ rows of cells to a string that involves $2$ rows of cells. We need to first derive elementary row operations to achieve this goal.

Since we are dealing with rows, the following notation will be useful.
\begin{defi}
For $j=1,\cdots, n-1$
\begin{equation}
\begin{aligned}
[:,j]^U &:= [1,j]^{UR} \Pi_{i=1}^{n-1}[i+1,j]^{UL}\\
[\bar{:},j]^U &:= [n,j]^{UL} \Pi_{i=1}^{n-1}[n-i,j]^{UR}.
\end{aligned}
\end{equation}
For $j=2,\cdots, n$
\begin{equation}
\begin{aligned}
[:,j]^D &:= [1,j]^{DR} \Pi_{i=1}^{n-1}[i+1,j]^{DL} \\
[\bar{:},j]^D &:= [n,j]^{DL}\Pi_{i=1}^{n-1} [n-i,j]^{DR}.
\end{aligned}
\end{equation}
For $j=1,\cdots, n$
\begin{align}
[:,j]^{-1} &:= \Pi_{i=1}^n [i,j]^{-1}.
\end{align}
\end{defi}

A useful subroutine is the so called internal reversal operation. This operation relates two different strings of cells that appear in the same row, with different orders. Furthermore, it also changes the choice of the clusters that define the universal recovery maps.
\begin{lem}\label{lemma:2D_internal_reversal}
(Internal reversal) For $j=2, \cdots, n-1$
\begin{equation}
[:,j]^U \aeqn [\bar{:},j]^D.
\end{equation}
\end{lem}
\begin{proof}
By Lemma \ref{lemma:2D_inheritance}, we have $[i,j]^{UR}[i+1,j]^{UL}\aeq [i+1,j]^{DL}[i,j]^{DR}$. By the local consistency condition,
\begin{equation}
[i+1,j]^{DL}[i,j]^{DR} \aeq [i+1,j]^{UR} [i,j]^{DR}.
\end{equation}
Also, one can move $[i,j]^{DR}$ towards the end of the string until it reaches $[i-1,j]$ by applying a manifest relation because $[i,j]\centernot{-} [i',j]$ for $i'\geq i+2$. We repeat this procedure until $i=n-2$ and then use $[n-1,j]^{UR}[n,j]^{UL} \aeq [n,j]^{DL} [n-1,j]^{DR}$. This completes the proof.
\end{proof}

\begin{lem}\label{lemma:2D_forward_row_contraction}
(Forward row contraction) For $j=1,\cdots, n-2$
\begin{equation}
[:, j]^U [:,j+1]^{D} [:,j]^{-1} \aeqn [:,j+1]^{U}.
\end{equation}
\end{lem}
\begin{proof}
By using the manifest relations, one can rearrange the cells so that the beginning part of the string is of the following form:
\begin{equation}
[1,j]^{UR}[2,j]^{UL}[1,j+1]^{DR}[2,j+1]^{DL} [1,j]^{-1}\cdots.
\end{equation}
By using Lemma \ref{lemma:2D_contraction_LD}, this is equivalent, up to an $O(\epsilon)$ trace distance, to the following state:
\begin{equation}
[2,j+1]^{DL} [2,j]^{UL}[1,j+1]^{DR} \cdots.
\end{equation}
By using Lemma \ref{lemma:2D_inheritance_vertical}, this is equivalent, up to an $O(\epsilon)$ trace distance, to the following state:
\begin{equation}
[2,j]^{UR}[2,j+1]^{DR}[1,j+1]^{DR}\cdots.
\end{equation}
By using the manifest relations, one can move $[1,j+1]^{DR}$ towards the end of the string. Then one can bring $[3,j]^{UL}$, $[3,j+1]^{DL}$, and $[2,j]^{-1}$ to the beginning part of the string. Noting that $[3,j] \centernot{-} [2,j+1]$, we can rearrange the beginning part of the string as
\begin{equation}
[2,j]^{UR}[3,j]^{UL}[2,j+1]^{DR}[3,j+1]^{DL} [1,j]^{-1}\cdots.
\end{equation}
Now we can repeat the procedure delineated above. That is, we use Lemma \ref{lemma:2D_contraction_LD} and then Lemma \ref{lemma:2D_inheritance_vertical}, and then rearrange the string using the manifest relations, and repeat the process.

At the end of the procedure, the beginning part of the string looks as follows:
\begin{equation}
[n,j+1]^{DL}[n,j]^{UL}\cdots.
\end{equation}
By invoking the inherited Markov condition, $[n,j+1]^{DL}[n,j]^{UL}[n,j]^{-1} \aeq [n,j+1]^{DL}$, thus completely removing all the cells on the $j$th row.

Meanwhile, what has been going on at the end of the string? We have been dumping $[i,j+1]^{UR}$ from $i=1$ to $n-1$, and consequently, they have been piling up at the end. The entire string is
\begin{equation}
[n,j+1]^{DL} (\Pi_{i=1}^{n-1}[n-i,j+1]^{DR}).
\end{equation}
Thus we have shown that
\begin{equation}
[:, j]^U [:,j+1]^{D} [:,j]^{-1}\aeqn [\bar{:}, j+1]^D.
\end{equation}
By invoking Lemma \ref{lemma:2D_internal_reversal} and applying the triangle inequality, the claim is proved.
\end{proof}

Next is an analogue of the cell exchange, which we call as the row exchange.
\begin{lem}\label{lemma:2D_row_exchange} (Row exchange) For $j=1, \cdots, n-1$
\begin{equation}
[:,j]^U[:,j+1]^D \aeqn [\bar{:}, j+1]^D [\bar{:}, j]^{U}
\end{equation}
\end{lem}
\begin{proof}
By using the manifest relations, one can rearrange the cells so that the beginning part of the string is of the following form:
\begin{equation}
[1,j]^{UR}[2,j]^{UL}[1,j+1]^{DR}[2,j+1]^{DL} \cdots.
\end{equation}
By Lemma \ref{lemma:2D_4cell_permutation}, this is equivalent, up to an $O(\epsilon)$ trace distance, to the following state:
\begin{equation}
[2,j+1]^{DL}[1,j+1]^{DR}[2,j]^{UL}[1,j]^{UR} \cdots.
\end{equation}
Since $[1,j+1] \centernot{-}[2,j]$, this is exactly equal to the following string:
\begin{equation}
[2,j+1]^{DL}[2,j]^{UL}[1,j+1]^{DR}[1,j]^{UR} \cdots.
\end{equation}
By using the manifest relations, one can move $[1,j+1]^{DR}[1,j]^{UR}$ to the end of the string.
\begin{equation}
[2,j+1]^{DL}[2,j]^{UL}\cdots [1,j+1]^{DR}[1,j]^{UR}.
\end{equation}
Now we use Lemma \ref{lemma:2D_inheritance_vertical} to convert the string into the following form with an $O(\epsilon)$ error:
\begin{equation}
[2,j]^{UR}[2,j+1]^{DR}\cdots [1,j+1]^{DR}[1,j]^{UR}.
\end{equation}
By using the manifest relations, one can rearrange the beginning part of the string as follows:
\begin{equation}
[2,j]^{UR}[3,j]^{UL}[2,j+1]^{DR}  [3,j+1]^{DL} \cdots  [1,j+1]^{DR}[1,j]^{UR}.
\end{equation}
By repeating this procedure $O(n)$ times, the string is transformed into the following form:
\begin{equation}
[n,j+1]^{DL}[n,j]^{UL} (\Pi_{i=1}^{n-1} [n-i,j+1]^{DR}[n-i,j]^{UR}),
\end{equation}
where in the last step we used $[n-1,j]^{UR}[n,j]^{UL}[n-1,j+1]^{DR}[n,j+1]^{DL} \aeq[n,j+1]^{DL}[n,j]^{UL}[n-1,j+1]^{DR}[n-1,j]^{UR}$.
Using the manifest relations, one can rearrange the string as
\begin{equation}
[n,j+1]^{DL}(\Pi_{i=1}^{n-1} [n-i,j+1]^{DR})[n,j]^{UL} (\Pi_{i=1}^{n-1} [n-i,j]^{UR}).
\end{equation}
Throughout this entire procedure we have used $O(n)$ derived relations, and thus the claim is proved.
\end{proof}

Now we prove an analogue of the backward contraction.
\begin{lem} \label{lemma:2D_backward_row_contraction}
(Backward row contraction) For $j=2,\cdots, n-1$,
\begin{equation}
[\bar{:}, j+1]^D[\bar{:},j]^U [:, j+1]^{-1} \aeqn [\bar{:}, j]^D.
\end{equation}
\end{lem}
\begin{proof}
The proof is similar to Lemma \ref{lemma:2D_forward_row_contraction}. The only difference is the order of the cells and the fact that we use Lemma \ref{lemma:2D_contraction_UR} instead of Lemma \ref{lemma:2D_contraction_LD}. By rotating the diagram by $\pi$, one should be able to see that the structure of the proof is exactly the same. Nevertheless, we explain each of the steps for concreteness. First, we use the manifest relations so that the beginning part of the string has the form of $[n,j+1]^{DL}[n-1,j+1]^{DR}[n,j]^{UL}[n-1,j]^{UR}[n,j+1]^{-1}$. We then apply Lemma \ref{lemma:2D_contraction_UR}, apply Lemma \ref{lemma:2D_inheritance_vertical} and manifest relations so that the same type of procedure can be applied repeatedly. After applying this procedure $O(n)$ times, we end up with the string $[:,j]^U$. By Lemma \ref{lemma:2D_internal_reversal}, this is close to $[:,i]^D$ up to a trace distance of $O(n\epsilon)$. Thus the claim is proved.
\end{proof}

The relations involving the rows are summarized in Table \ref{table:row_relations}.
\begin{table}[h]
\begin{tabular}{c c c c}
Name & Relations & Derivation & Appearance \\ \hline
Internal reversal & $[:,j]^US \aeqn [\bar{:},j]^DS$ & Lemma \ref{lemma:2D_inheritance}, Consistency & Lemma \ref{lemma:2D_internal_reversal} \\ \hline
Forward row contraction& $[:,j]^U [:,j+1]^{D} [:,j]^{-1}S \aeqn [:,j+1]^{U}S$ & Lemma \ref{lemma:2D_inheritance_vertical}, \ref{lemma:2D_contraction_LD}, \ref{lemma:2D_internal_reversal} & Lemma \ref{lemma:2D_forward_row_contraction} \\ \hline
Row exchange&$[:,j]^U[:,j+1]^DS \aeqn [\bar{:}, j+1]^D [\bar{:} j]^{U}S$ & Lemma \ref{lemma:2D_inheritance_vertical}, \ref{lemma:2D_4cell_permutation} & Lemma \ref{lemma:2D_row_exchange} \\ \hline
Backward row contraction&$[\bar{:}, j+1]^D[\bar{:},j]^U [:, j+1]^{-1}S \aeqn [\bar{:}, j]^DS$ & Lemma \ref{lemma:2D_inheritance_vertical}, \ref{lemma:2D_contraction_UR} & Lemma \ref{lemma:2D_backward_row_contraction} \\ \hline
\end{tabular}
\caption{Relations involving rows of cells. The string $S$ can be arbitrary, provided that the strings on both sides of the relations are well-defined.\label{table:row_relations}}
\end{table}

Now we arrive at our temporary conclusion: that the reduced density matrix of our proposed state can be represented by a string of cells supported on two rows.
\begin{prop}\label{prop:two_row}
\begin{equation}
[:,1]^U(\Pi_{i=1}^{n-1} [:,i+1]^{D}) (\Pi_{i=1}^{m-1} [:,i]^{-1}) (\Pi_{i=m+2}^{n}[:,i]^{-1}) \underset{O(n^2\epsilon)}{\approx} [:,m]^{U}[:,m+1]^D. \label{eq:two_row}
\end{equation}
\end{prop}
\begin{proof}
By applying manifest relations, the string can be rearranged so that the beginning part is of the following form:
\begin{equation}
[:,1]^U[:,2]^D[:,1]^{-1}\cdots,
\end{equation}
which is equivalent to $[:,2]^U\cdots $ up to a trace distance of $O(n\epsilon)$ by using the forward contraction; see Lemma \ref{lemma:2D_forward_row_contraction}. After repeating this procedure, we arrive at the following string:
\begin{equation}
[:,m]^U (\Pi_{i=1}^{n-m} [:,m+i]^D) (\Pi_{i=m+2}^{n}[:,i]^{-1}),
\end{equation}
with a trace distance bounded from above by $O(nm\epsilon)$.
Now we apply the row exchange operation. Specifically,
\begin{equation}
\begin{aligned}
[:,m]^U [:,m+1]^D &\aeqn [\bar{:},m+1]^D[\bar{:},m]^U\\
&\aeqn [:,m+1]^U [\bar{:},m]^U,
\end{aligned}
\end{equation}
where we used Lemma \ref{lemma:2D_row_exchange} in the first line and Lemma \ref{lemma:2D_internal_reversal} in the second line. By using the manifest relation, one can move $[\bar{:},m]^U$ towards the end of the string. Specifically,
\begin{equation}
[:,m+1]^U [:,m+2]^D\cdots [\bar{:},m]^U(\Pi_{i=m+2}^{n}[:,i]^{-1}).
\end{equation}
By repeatedly applying this procedure $O(n-m)$ times, the string can be converted to the following form:
\begin{equation}
[\bar{:},n]^D(\Pi_{i=1}^{n-m}[\bar{:},n-i]^U)(\Pi_{i=m+2}^{n}[:,i]^{-1}).
\end{equation}
After applying the backward row contraction(cf. Lemma \ref{lemma:2D_backward_row_contraction}) $O(n-m)$ times, the string is converted into the following form:
\begin{equation}
[\bar{:},m+1]^D [\bar{:},m]^U,
\end{equation}
which can be then converted to $[:,m]^U[:,m+1]^D$ by Lemma \ref{lemma:2D_row_exchange}. All the trace distances that are incurred throughout this procedure is bounded by $O(n(n-m)\epsilon)$. Thus the trace distance between the two states in Eq.\ref{eq:two_row} is bounded from above by $O(n^2 \epsilon)$.
\end{proof}

\subsection{Supercell relations\label{section:supercell_relations}}
The main takeaway message so far should be the fact that the reduced density matrix of the proposed state on any two contiguous rows is approximately equal to a string that only consists of cells on these rows. Since all of the clusters are supported on at most two rows, proving Theorem \ref{thm:2D} now amounts to proving that the reduced density matrix of these states over each of the clusters in the two rows is (approximately) equal to the given marginals. Unsurprisingly, the main theme of Theorem \ref{thm:1D} and Proposition \ref{prop:two_row} carries over.
The key observation is that one can rearrange the string $[:,m]^U[:,m+1]^D$ into the following form:
\begin{equation}
[1,m]^{UR}[1,m+1]^{DR}\Pi_{i=1}^{n-1} [1+i,m]^{UL}[1+i,m+1]^{DL}.
\end{equation}
Now we can view $[i,m]\cup [i,m+1]$ as one supercell. The forward contraction, cell exchange, and the backward contraction involves these supercells.\footnote{We do not have a table that summarizes these relations. Unfortunately the strings are too long to be contained in a table.}
\begin{lem}\label{lemma:2D_forward_supercell_contraction}
(Forward supercell contraction) For $i=1,\cdots,n-2$
\begin{equation}
[i,m]^{UR}[i,m+1]^{DR}[i+1,m]^{UL}[i+1,m+1]^{DL} [i,m]^{-1}[i,m+1]^{-1} \aeq [i+1,m]^{UR}[i+1,m+1]^{DR}.
\end{equation}
\end{lem}
\begin{proof}
Note that all the clusters that appear on the left hand side are in fact the same cluster. Let us denote this cluster as $A$ and the cluster on its right as $B$. To be clear, $B$ consists of cells $[i+1,m], [i+1,m+1], [i+2,m],$ and $[i+2,m+1]$. By using the local Markov conditions, one can show that $[i,m]^{UR}[i,m+1]^{DR}[i+1,m]^{UL}[i+1,m+1]^{DL}$ is close to the marginal over $A$ up to a trace distance of $O(\epsilon)$.\footnote{This is the same argument that was used in Lemma \ref{lemma:2D_4cell_permutation}, \ref{lemma:2D_contraction_LD}, and \ref{lemma:2D_contraction_UR}.} By the local consistency condition, the reduced density matrix of this marginal over $[i+1,m]\cup [i+1,m+1]$ is close to the reduced density matrix of the same region for $\rho^{\bar{B}}$. By the inherited local Markov condition, the reduced density matrix over this region for $\rho^{\bar{B}}$ is $O(\epsilon)$ close to $[i+1,m]^{UR}[i+1,m+1]^{DR}$ in trace distance. This completes the proof.
\end{proof}

\begin{lem}\label{lemma:2D_supercell_exchange}
(Supercell exchange) For $i=1,\cdots, n-2$,
\begin{equation}
[i,m]^{UR}[i,m+1]^{DR}[i+1,m]^{UL}[i+1,m+1]^{DL} \aeq [i+1,m]^{UR}[i+1,m+1]^{DR}[i,m+1]^{DR} [i,m]^{UR}
\end{equation}
\end{lem}
\begin{proof}
We first apply Lemma \ref{lemma:2D_4cell_permutation} to show that
\begin{equation}
[i,m]^{UR}[i,m+1]^{DR}[i+1,m]^{UL}[i+1,m+1]^{DL} \aeq [i+1,m+1]^{DL}[i+1,m]^{UL}[i,m+1]^{DR}[i,m]^{UR}.
\end{equation}
Then we apply Lemma \ref{lemma:2D_inheritance_vertical} to show that
\begin{equation}
[i+1,m+1]^{DL}[i+1,m]^{UL}[i,m+1]^{DR}[i,m]^{UR} \aeq [i+1,m]^{UR}[i+1,m+1]^{DR}[i,m+1]^{DR}[i,m]^{UR}.
\end{equation}
This completes the proof.
\end{proof}

\begin{lem}\label{lemma:2D_backward_supercell_contraction}
(Backward supercell contraction) For $i=2,\cdots,n-1$,
\begin{equation}
[i+1,m+1]^{DL}[i+1,m]^{UL}[i,m+1]^{DR}[i,m]^{UR} [i+1,m+1]^{-1}[i+1,m]^{-1} \aeq [i,m+1]^{DL}[i,m]^{UL}
\end{equation}
\end{lem}
\begin{proof}
Note that all the clusters that appear on the left hand side are in fact the same cluster. Let us denote this cluster as $A$ and the cluster on its left as $B$. To be clear, $B$ consists of cells $[i-1,m], [i-1,m+1], [i,m],$ and $[i,m+1]$. By using the local Markov conditions, one can show that $[i+1,m+1]^{DL}[i+1,m]^{UL}[i,m+1]^{DR}[i,m]^{UR}$ is close to the marginal over $A$ up to a trace distance of $O(\epsilon)$. By the local consistency condition, the reduced density matrix of this marginal over $[i,m+1]\cup [i,m]$ is close to the reduced density matrix of the same region for $\rho^{\bar{B}}$. By the inherited local Markov condition, the reduced density matrix over this region for $\rho^{\bar{B}}$ is $O(\epsilon)$ close to $[i,m+1]^{DL}[i,m]^{UL}$ in trace distance. This completes the proof.
\end{proof}

Now we are in a position to complete the proof of Theorem \ref{thm:2D}. Without loss of generality, pick some $i$ between $1$ and $n,$ and trace out all the cells but the ones with a $x$ coordinate of $i$ and $i+1$ for the state $[:,m]^U[:,m+1]^D$. By applying the forward supercell contraction, one can remove all the cells with $x$ coordinate less than $i$. By sequentially applying the supercell exchange, the order of the cells are reversed.\footnote{At the last step, one should use $[n-1,m]^{UR}[n-1,m+1]^{DR}[n,m]^{UL}[n,m+1]^{DL} \aeq[n,m+1]^{DL}[n,m]^{UL}[n-1,m+1]^{DR}[n-1,m]^{UR}$(Lemma \ref{lemma:2D_4cell_permutation}) as opposed to Lemma \ref{lemma:2D_supercell_exchange}.} Then one can sequentially apply the backward supercell contraction. What remains is a state over $4$ cells $[i,m],[i+1,m],[i,m+1]$, and $[i+1,m+1]$, which is represented by the following string:
\begin{equation}
[i+1,m+1]^{DL} [i,m+1]^{DR}[i+1,m]^{UL}  [i,m]^{UR}.
\end{equation}
As already discussed in Lemma \ref{lemma:2D_4cell_permutation}, \ref{lemma:2D_contraction_LD}, and \ref{lemma:2D_forward_supercell_contraction}, this is close to the marginal over $[i,m]\cup [i,m+1]\cup [i+1,m] \cup [i+1,m+1]$ up to a trace distance of $O(\epsilon).$ Combining all the trace distance estimates, we conclude that for all $i=1,\cdots, n-1$ and $m=1,\cdots, n-1$ the reduced density matrix of the proposed state(cf. Eq.\ref{eq:proposed_state_2D}) over this regions is $O(n^2\epsilon)$ close to the given marginals, thus completing the proof.

\section{Conclusion\label{section:conclusion}}
We introduced a special class of marginals, the so called Markovian marginals. This construction leads to a large family of nontrivial solutions to the quantum marginal problem. This was possible because the local Markov condition ensures a set of nontrivial relations involving the polymorphic extensions and contractions. By combining these relations, we were able to prove the consistency of the given marginals. The conditions that ensure the consistency of the Markovian marginal is physically well-motivated and reasonable, as we explained in Section \ref{section:summary}. By assuming translational invariance, a single density matrix obeying these conditions defines a local reduced density matrix of an infinite system in the $\epsilon \to 0$ limit. It would be interesting to minimize the energy in the space of such marginals and study the thermodynamic properties of interacting quantum many-body systems.

There are several issues that we did not discuss in this paper. This includes the maximum global entropy that is compatible with the given Markovian marginal. At least for the examples discussed in this paper, there is an exact formula with rigorous stability bound. Also, it is possible to compute the long-range correlation functions efficiently. An algorithm for minimizing the energy of a Markovian marginal is also an important problem. These studies will appear elsewhere.

Let us comment on our choice of partitions and clusters. Our sole intention was to simplify certain aspects of the proof and the formulation of the statement.\footnote{Without these objects one would have had to specify the reduced density matrix as well as the domain and the codomain of the universal recovery maps for every lemma. } One could have completely abandoned this notion and just specify the local Markov conditions and local consistency conditions explicitly, but we found such a formulation to be a bit awkward. This is not to say that the present formulation is without any flaws. For example, the local Markov conditions corresponding to FIG.\ref{fig:CMI2D_3-3}, \ref{fig:CMI2D_3-6}, \ref{fig:CMI2D_4-1}, and \ref{fig:CMI2D_4-3} were never used. These conditions are implied by Eq.\ref{eq:EEscaling} anyway, so we did not bother to point out that they are irrelevant for the proof. Also, there are Markovian marginals that lie strictly oustide of the framework developed in this paper. They are likely to be more practical, because the conditions that ensure the consistency for those Markovian marginals is weaker. In light of these facts, this paper should be viewed as an existence proof that nontrivial class of physically relevant solutions to the quantum marginal problem exists. Finding a minimal set of conditions that ensures the consistency of the marginals is an important problem that warrants a further study.

We should also point out a subtlety in our claim: that our solution is applicable to topologically ordered states. This is only true in the sense that our solution is applicable to arbitrarily large regions of a topologically ordered system on an infinite plane. Our solution is not applicable to topologically ordered states on a closed system. In fact, the following argument shows that such a solution is unlikely to exist. Suppose, for example, that there exists a Markovian marginal whose marginals consist of bounded regions of the toric code\cite{Kitaev1997} ground state such that it ensures the existence of a consistent global state. Then one can consider a different Markovian marginal such that one of the marginals is replaced with a reduced density matrix of the toric code that contains a nontrivial topological charge in that region. For such marginals, the charge of the entire system cannot add up to be a trivial charge, and thus cannot be consistent. However, the formulation of the Markovian marginal is such that the requisite constraints are specified only in terms of the local consistency condition and the entanglement entropy. The value of the entanglement entropy, for the case of the toric code, is independent of the topological charge that is enclosed in the region.\cite{Kitaev2006} Therefore, even for such a choice of marginal, the local Markov condition would be satisfied. Also, by placing the topological charge on a region that is sufficiently far away from the overlaps between the marginals, one can ensure that the local consistency condition remains intact. This means that, if there exists a Markovian marginal that consists of marginals of the toric code such that one can prove its consistency in a manner we did in this paper, the same proof must go through for a set of marginals for which a consistent global state cannot exist.\footnote{I thank Sergey Bravyi for pointing out this fact.}

An intriguing question is whether one can design a suitable family of Markovian marginal to study large molecules. Markovian marginal seems to be sufficiently flexible that it may allow such a possibility. Obviously, the answer to this question would depend on the details of the molecule and the structure of the correlations that are present in such systems. It will be illuminating to study the local Markov condition over different orbitals in small molecules and deduce a potential pattern that may be applicable to larger molecules. On a more technical side, one would need to judiciously reformulate the definition of polymorphic extensions and contractions for CAR algebra. The existence of the universal recovery map for fermions should follow from Ref.\cite{Junge2015}.

Lastly, we emphasize that a Markovian marginal is not a graphical model,\cite{Pearl1988} even if we restrict ourselves to classical probability distributions. While both the graphical model and a Markovian marginal is formulated in terms of Markov conditions, the Markov conditions for the Markovian marginal are local, whereas the Markov conditions for graphical models are global. As such, a Markovian marginal, even in the classical regime, should be thought as an alternative generalization of a Markov chain. It will be interesting to study its applications in the domain of variational Bayesian methods.
\begin{acknowledgements}
I thank Sergey Bravyi for the helpful discussions. Part of this work was done during a workshop on quantum marginals and numerical ranges at the University of Guelph. My research at Perimeter Institute was supported by the Government of Canada through Industry Canada and by the Province of Ontario through the Ministry of Economic Development and Innovation.
\end{acknowledgements}
\bibliography{bib}

\end{document}